\documentclass[11pt,oneside,article,a4paper]{memoir}
\usepackage{amsmath}
\usepackage{amsfonts}
\usepackage{amssymb}
\usepackage{bm}
\usepackage[english]{babel}
\usepackage[latin1]{inputenc}
\usepackage[T1]{fontenc}
\usepackage{graphicx}
\usepackage{mathtools}
\usepackage{mathrsfs}
\usepackage{cancel}
\usepackage[amsmath,thmmarks]{ntheorem}
\usepackage{color}
\usepackage{rotating}
\newsubfloat{figure}

\DeclareMathOperator*{\esssup}{ess\,sup}
\newcommand*{\longhookrightarrow}{\ensuremath{\lhook\joinrel\relbar\joinrel\rightarrow}}

\title{{\scshape\textbf{Existence of a Unique Local Solution to the Many-body Maxwell-Schrödinger Initial Value Problem}}}
\author{Kim Petersen}
\date{\small Department of Mathematical Sciences, University of Copenhagen,\\
Universitetsparken 5, DK-2100 Copenhagen, Denmark, Email: kp@math.ku.dk\\
\vspace{0.2cm} \copyright\ 2013 by the author}

\abstractrunin
\abslabeldelim{.}

\theoremstyle{plain}
\theoremseparator{.}
\newtheorem{saetning}{Theorem}
\newtheorem{korollar}[saetning]{Corollary}
\newtheorem{lemma}[saetning]{Lemma}

\theoremstyle{plain}
\theorembodyfont{\normalfont}
\theoremseparator{.}
\newtheorem{bemaerkning}[saetning]{Remark}

\theoremstyle{nonumberplain}
\theorembodyfont{\normalfont}
\theoremsymbol{\ensuremath{\square}}
\theoremseparator{.}
\newtheorem{proof}{Proof}

\setsecheadstyle{\normalfont\scshape}
\begin{document}
\maketitle
\begin{abstract}
We study the many-body problem of charged particles interacting with their self-generated electromagnetic field. We model the dynamics of the particles by the many-body Maxwell-Schrödinger system, where the particles are treated quantum mechanically and the electromagnetic field is a classical quantity. We prove the existence of a unique local in time solution to this nonlinear initial value problem using a contraction mapping argument.
\end{abstract}

\chapter{Introduction}
The three-dimensional many-body Maxwell-Schrödinger system in Coulomb gauge is a system of partial differential equations that models the dynamics of several charged point particles interacting via their self-generated electromagnetic fields -- in Gaussian units it reads
\begin{align}
\begin{split}\label{MSP}
\Box\bm{A}&=\frac{4\pi}{c}\sum_{j=1}^N P\bm{J}_j[\psi,\bm{A}],\\
i\hbar \partial_t\psi&=\Bigl(\sum_{j=1}^N \frac{1}{2m_j}\nabla_{j,\bm{A}}^2+\sum_{1\leq j<k\leq N} \frac{Q_jQ_k}{|\bm{x}_j-\bm{x}_k|}+\mathscr{E}_{\mathrm{EM}}[\bm{A},\partial_t\bm{A}]\Bigr)\psi,
\end{split}
\end{align}
where $\hbar>0$ is the reduced Planck constant, $c>0$ is the speed of light, $N\in\mathbb{N}$ is the number of particles, $m_1,\ldots,m_N>0$ are the particles' respective masses, $Q_1,\ldots,Q_N\in\mathbb{R}$ are their charges, $\psi(t):\mathbb{R}^{3N}\to\mathbb{C}$ is the wave function, $\bm{A}(t):\mathbb{R}^3\to\mathbb{R}^3$ is the vector potential, $\nabla_{j,\bm{A}}=i\hbar \nabla_{\bm{x}_j}+\frac{Q_j}{c} \bm{A}(\bm{x}_j)$ is the covariant derivative with respect to $\bm{A}$ acting on the $j$'th particle, $\Box=\frac{1}{c^2}\partial_t^2-\Delta$ denotes the d'Alembertian, $P=1-\nabla\mathrm{div}\Delta^{-1}$ is the Helmholtz projection, $\mathscr{E}_{\mathrm{EM}}[\bm{A},\partial_t\bm{A}](t)$ is the field energy
\begin{align*}
\mathscr{E}_{\mathrm{EM}}[\bm{A},\partial_t\bm{A}](t)=\frac{1}{8\pi}\int_{\mathbb{R}^3}\Bigl(|\nabla\times\bm{A}(t)(\bm{y})|^2+\Bigl|\frac{1}{c}\partial_t\bm{A}(t)(\bm{y})\Bigr|^2\Bigr)\,\mathrm{d}\bm{y}
\end{align*}
and $\bm{J}_j[\psi,\bm{A}](t)$ denotes the $j$'th particle's probability current density
\begin{align*}
\bm{J}_j[\psi,\bm{A}](t):\mathbb{R}^3\ni \bm{x}_j\mapsto\Bigl(-\frac{Q_j}{m_j}\mathrm{Re}\int_{\mathbb{R}^{3(N-1)}}\hspace{-0.5cm}\overline{\psi}(t)(\bm{x})\nabla_{j,\bm{A}}\psi(t)(\bm{x})\,\mathrm{d}\bm{x}_j'\Bigr)\in\mathbb{R}^3
\end{align*}
where $\bm{x}=(\bm{x}_1,\ldots,\bm{x}_N)$ and $\bm{x}_j'=(\bm{x}_1,\ldots,\bm{x}_{j-1},\bm{x}_{j+1},\ldots,\bm{x}_N)$. That we have chosen Coulomb gauge means that the magnetic vector potential $\bm{A}$ should satisfy
\begin{align}
\mathrm{div}\bm{A}(t)=0.\label{Coulombgauge}
\end{align}
As we will see later we might just as well study the system
\begin{align}
\begin{split}\label{MSPmaerke}
\Box\bm{A}&=\frac{4\pi}{c}\sum_{j=1}^N P\bm{J}_j[\psi,\bm{A}],\\
i\hbar \partial_t\psi&=\Bigl(\sum_{j=1}^N \frac{1}{2m_j}\nabla_{j,\bm{A}}^2+\sum_{1\leq j<k\leq N} \frac{Q_jQ_k}{|\bm{x}_j-\bm{x}_k|}\Bigr)\psi,
\end{split}
\end{align}
where the field energy-term $\mathscr{E}_{\mathrm{EM}}[\bm{A},\partial_t\bm{A}]$ is no longer present in the Schrödinger equation. In the literature the $d$-dimensional Maxwell-Schrödinger system often refers to the coupled equations
\begin{align}
\begin{split}\label{MSCN0}
-\Delta_d\varphi-\frac{1}{c}\partial_t\mathrm{div}_d\bm{A}&=4\pi Q|\psi|^2,\\
\Box_d\bm{A}+\nabla_d \Bigl(\frac{1}{c}\partial_t\varphi+\mathrm{div}_d\bm{A}\Bigr)&=-\frac{4\pi}{c}\frac{Q}{m}\mathrm{Re}\Bigl(\overline{\psi}\Bigl(i\hbar\nabla_d+\frac{Q}{c}\bm{A}\Bigr)\psi\Bigr),\\
i\hbar\partial_t\psi&=\Bigl(\frac{1}{2m}\Bigl(i\hbar\nabla_d+\frac{Q}{c}\bm{A}\Bigr)^2+Q\varphi\Bigr)\psi,
\end{split}
\end{align}
with unknowns $\psi(t):\mathbb{R}^d\to\mathbb{C}$, $\bm{A}(t):\mathbb{R}^d\to\mathbb{R}^d$, $\varphi(t):\mathbb{R}^d\to\mathbb{R}$ and a hopefully obvious notation. If $d=3$ and $\bm{A}$ satisfies the Coulomb gauge-condition \eqref{Coulombgauge} the system \eqref{MSCN0} reads
\begin{align}
\begin{split}\label{MSCN}
\Box\bm{A}&=P\Bigl(-\frac{4\pi}{c}\frac{Q}{m}\mathrm{Re}\Bigl(\overline{\psi}\Bigl(i\hbar\nabla+\frac{Q}{c}\bm{A}\Bigr)\psi\Bigr)\Bigr),\\
i\hbar\partial_t\psi&=\Bigl(\frac{1}{2m}\Bigl(i\hbar \nabla+\frac{Q}{c}\bm{A}\Bigr)^2+Q^2\bigl(|\bm{x}|^{-1}*|\psi|^2\bigr)\Bigr)\psi,
\end{split}
\end{align}
which only differs from the $(N=1)$-case of \eqref{MSPmaerke} by the presence of the nonlinear term $Q^2(|\bm{x}|^{-1}*|\psi|^2)\psi$ in the Schrödinger equation. This term comes from the particles' Coulomb self-interactions. From a physical point of view it is wrong to include self-interactions in this context. In fact, the system \eqref{MSCN} may be considered as a mean field approximation to the many-body description \eqref{MSPmaerke}. In the $(c\to\infty)$-limit the second equation in \eqref{MSPmaerke} reduces to the standard many-body Coulomb problem
\begin{align*}
i\hbar\partial_t\psi=\Bigl(-\sum_{j=1}^N \frac{\hbar^2}{2m_j}\Delta_{\bm{x}_j}+\sum_{1\leq j<k\leq N}\frac{Q_jQ_k}{|\bm{x}_j-\bm{x}_k|}\Bigr)\psi
\end{align*}
which is the basis of almost all work done in quantum chemistry. On the other hand, the system \eqref{MSCN} reduces in the $(c\to\infty)$-limit to a mean field approximation, which at best would be good in the large $N$ limit. The system \eqref{MSCN0} has been studied (up to different choices of units) by several authors, both when expressed in the Coulomb gauge \cite{Ta,Gin1,Gin4,Gin2,Gin3,GNS,NW,NWG,Shi,Tsu1}, the Lorenz gauge \cite{GNS,NT,NW,NWG,Wada} and the temporal gauge \cite{GNS,NW,NWG}. In \cite{NT}, Nakamitsu and Tsutsumi prove local well-posedness in sufficiently regular Sobolev spaces of the $d$-dimensional Maxwell-Schrödinger initial value problem -- for $d\in\{1,2\}$ they also show global existence of the solution. Tsutsumi shows in \cite{Tsu1} that for $d=3$ the problem has a global solution for a certain set of final states (i.e. data given at $t=+\infty$) and studies the asymptotic behavior of such a solution. In \cite{GNS}, Guo, Nakamitsu and Strauss prove global solvability of the three-dimensional system in Coulomb gauge (but not uniqueness of the solution) for initial data $(\psi(0),\bm{A}(0),\partial_t\bm{A}(0))$ in the space of $H^1\times H^1\times L^2$-functions satisfying $\mathrm{div}\bm{A}(0)=\mathrm{div}\partial_t\bm{A}(0)=0$. Using techniques on which the arguments in the present paper are based, Nakamura and Wada \cite{NW} prove local well-posedness of the three-dimensional problem in Sobolev spaces of sufficient regularity, expanding significantly on the previously known results -- in \cite{NWG} they even prove global existence of unique solutions. Bejenaru and Tataru \cite{Ta} prove global well-posedness in the energy-space of the three-dimensional initial value problem and in the recent paper \cite{Wada}, Wada proves unique solvability in the energy space of the two-dimensional analogue. The scattering theory for \eqref{MSCN} has also been studied by several authors -- see the papers by Tsutsumi \cite{Tsu1}, Shimomura \cite{Shi} as well as Ginibre and Velo \cite{Gin1,Gin4,Gin2,Gin3}. It seems that the solvability of the system \eqref{MSP} has not yet been studied and the known results concerning \eqref{MSCN0} are not directly applicable to this system due to the presence of the Coulomb singularities $\frac{1}{|\bm{x}_j-\bm{x}_k|}$ in \eqref{MSP}. The aim of this paper is to prove the unique existence of a local solution to \eqref{MSP} as expressed in the following main theorem.
\begin{saetning}\label{main}
For all $(\psi_0,\bm{A}_0,\bm{A}_1)\in H^2(\mathbb{R}^{3N})\times H^{\frac{3}{2}}(\mathbb{R}^3;\mathbb{R}^3)\times H^{\frac{1}{2}}(\mathbb{R}^3;\mathbb{R}^3)$ with $\mathrm{div}\bm{A}_0=\mathrm{div}\bm{A}_1=0$ there exist a number $T>0$ and a unique solution $(\psi,\bm{A})\in C([0,T];H^2(\mathbb{R}^{3N}))\times \bigl(C([0,T];H^{\frac{3}{2}}(\mathbb{R}^3;\mathbb{R}^3))\cap C^1([0,T];H^{\frac{1}{2}}(\mathbb{R}^3;\mathbb{R}^3))\bigr)$ to \eqref{MSP} such that $\mathrm{div}\bm{A}(t)=0$ for all $t\in[0,T]$ and the initial conditions
\begin{align}
\psi(0)=\psi_0, \bm{A}(0)=\bm{A}_0\textrm{ and }\partial_t\bm{A}(0)=\bm{A}_1\label{initialcondi}
\end{align}
are satisfied.
\end{saetning}
\begin{bemaerkning}
In this paper, we consider all of the charged particles as being spinless. Let us just mention that by thinking of the particles as having spin and by including the interaction between this spin and the electromagnetic field in the kinetic energy operator we are led to another interesting system of partial differential equations: The many-body Maxwell-Pauli system. For now, let us just write up the one-body Maxwell-Pauli system -- in Coulomb gauge it reads
\begin{align}
\begin{split}\label{MP1}
\Box\bm{A}&=\frac{4\pi}{c}P\bm{\mathcal{J}}[\psi,\bm{A}],\\
i\hbar\partial_t\psi&=\left(\frac{1}{2m}\Bigl(i\hbar\nabla+\frac{Q}{c}\bm{A}\Bigr)^2-\frac{\hbar Q}{2mc}\bm{\sigma}\cdot\nabla\times\bm{A}+\mathscr{E}_{\mathrm{EM}}[\bm{A},\partial_t\bm{A}]\right)\psi,
\end{split}
\end{align}
where the probability current density $\bm{\mathcal{J}}[\psi,\bm{A}](t):\mathbb{R}^3\to\mathbb{R}^3$ is given by
\begin{align*}
\bm{\mathcal{J}}[\psi,\bm{A}](t)(\bm{x})&=-\frac{Q}{m}\mathrm{Re}\Bigl\langle\psi(t)(\bm{x}),\Bigl(i\hbar\nabla+\frac{Q}{c}\bm{A}\Bigr)\psi(t)(\bm{x})\Bigr\rangle_{\mathbb{C}^2}\\
&\hspace{1.7cm}+\frac{\hbar Q}{2m}\nabla\times\bigl\langle\psi(t)(\bm{x}),\bm{\sigma}\psi(t)(\bm{x})\bigr\rangle_{\mathbb{C}^2}
\end{align*}
and $\bm{\sigma}$ is the vector with the Pauli matrices
\begin{align*}
\sigma^1=\begin{pmatrix}0&1\\1&0\end{pmatrix},\quad \sigma^2=\begin{pmatrix}0&-i\\i&0\end{pmatrix}\quad\textrm{and}\quad \sigma^3=\begin{pmatrix}1&0\\0&-1\end{pmatrix}
\end{align*}
as components. The techniques used in this paper to treat the many-body Maxwell-Schrödinger system do not seem to be immediately adaptable to the Maxwell-Pauli system and so the existence of a local solution to the initial value problem corresponding to \eqref{MP1} is an open problem.
\end{bemaerkning}
\begin{bemaerkning}
Suppose that $m_1=\cdots=m_N$ and $Q_1=\cdots=Q_N$ so that the $N$ particles are indistinguishable and consider an initial state $\psi_0$ where either all of the particles are bosonic ($s=0$) or all of the particles are fermionic ($s=1$). If $e_{\ell n}:\mathbb{R}^{3N}\to \mathbb{R}^{3N}$ is the coordinate exchange map given by
\begin{align*}
e_{\ell n}(\bm{x}_1,\ldots,\bm{x}_\ell,\ldots,\bm{x}_n,\ldots,\bm{x}_N)=(\bm{x}_1,\ldots,\bm{x}_n,\ldots,\bm{x}_\ell,\ldots,\bm{x}_N)
\end{align*} 
this means that $\psi_0=(-1)^s\psi_0\circ e_{\ell n}$ for all $\ell,n\in\{1,\ldots,N\}$ with $\ell<n$. With $(\psi,\bm{A})$ denoting the solution to \eqref{MSP}+\eqref{initialcondi} whose existence is established in Theorem \ref{main} one can easily verify that $t\mapsto\bigl((-1)^s\psi(t)\circ e_{\ell n},\bm{A}(t)\bigr)$ solves \eqref{MSP}+\eqref{initialcondi} too. But then the uniqueness result of Theorem \ref{main} implies that the identity $\psi(t)=(-1)^s\psi(t)\circ e_{\ell n}$ holds at all times $t$ of existence so in other words the particles will continue to obey the same particle statistics as they did in the initial state.
\end{bemaerkning}
The paper is organized as follows. We will end this introduction by establishing some notation and in Section \ref{Motivation} we (formally) motivate the model \eqref{MSP}. In Section \ref{proof} we take the first steps towards proving Theorem \ref{main} -- the basic strategy for obtaining the existence part of the theorem will be to find a fixed point for the solution mapping associated with a certain linearization of the many-body Maxwell-Schrödinger system. The linear equations constituting this linearization are studied in Sections \ref{SchEq} and \ref{KleinGo} -- more specifically, the many-body Schrödinger equation is studied in Section \ref{SchEq} by means of a result by Kato \cite{Kato1,Kato} and in Section \ref{KleinGo} we recall a result developed by Brenner \cite{Brenner}, Strichartz \cite{Strichartz}, Ginibre and Velo \cite{Gin-1,Gin0} concerning the Klein-Gordon equation. Finally, we prove existence of the desired solution in Section \ref{Contraction} and the uniqueness part is proven in Section \ref{Uniqueness}.

As can be seen from the statement of Theorem \ref{main} the values of the time variable will vary in some closed interval $\mathcal{I}_T=[0,T]$ where $T>0$. For some given reflexive Banach space $(\mathcal{X},\|\cdot\|_{\mathcal{X}})$ we will let $C(\mathcal{I}_T;\mathcal{X})$ denote the space of continuous mappings $\mathcal{I}_T\to \mathcal{X}$ and $C^1(\mathcal{I}_T;\mathcal{X})$ will denote the subspace of maps $\psi\in C(\mathcal{I}_T;\mathcal{X})$ whose strong derivative
\begin{align*}
\partial_t\psi(t)=\begin{dcases}\lim_{h\to 0^+}\frac{\psi(t+h)-\psi(t)}{h}&\textrm{for }t=0\\\lim_{h\to 0\phantom{^+}}\frac{\psi(t+h)-\psi(t)}{h}&\textrm{for }t\in(0,T)\\\lim_{h\to 0^-}\frac{\psi(t+h)-\psi(t)}{h}&\textrm{for }t=T\end{dcases}
\end{align*}
is well defined and continuous everywhere in $\mathcal{I}_T$. For $p\in[1,\infty]$ we let $L^p(\mathcal{I}_{T};\mathcal{X})$ denote the space of (equivalence classes of) strongly Lebesgue-measurable functions $\psi:\mathcal{I}_{T}\to \mathcal{X}$ with the property that
\begin{align*}
\|\psi\|_{L_T^p\mathcal{X}}=\begin{dcases}\Bigl(\int_{\mathcal{I}_{T}}\|\psi(t)\|_{\mathcal{X}}^p\,\mathrm{d}t\Bigr)^{\frac{1}{p}}&\textrm{if } 1\leq p<\infty\\\esssup_{t\in\mathcal{I}_{T}}\|\psi(t)\|_{\mathcal{X}}&\textrm{if }p=\infty\end{dcases}
\end{align*}
is finite. Equipping $L^p(\mathcal{I}_{T};\mathcal{X})$ with the norm $\|\cdot\|_{L_T^p\mathcal{X}}$ results in a Banach space. Just as in the case where $\mathcal{X}=\mathbb{C}$ any given $\psi\in L^p(\mathcal{I}_{T};\mathcal{X})$ can be identified with the $\mathcal{X}$-valued distribution that sends $f\in C_0^\infty(\mathcal{I}_{T}^{\circ})$ into the Bochner integral $\int_{\mathcal{I}_{T}} \psi(t)f(t)\,\mathrm{d}t\in \mathcal{X}$; thus, it makes sense to consider the space $W^{1,p}(\mathcal{I}_{T};\mathcal{X})$ of $L^p(\mathcal{I}_{T};\mathcal{X})$-functions with distributional derivative $\partial_t\psi$ in $L^p(\mathcal{I}_{T};\mathcal{X})$, which is a Banach space when endowed with the norm
\begin{align*}
\|\psi\|_{W_T^{1,p}\mathcal{X}}=\bigl(\|\psi\|_{L_T^p\mathcal{X}}^2+\|\partial_t\psi\|_{L_T^p\mathcal{X}}^2\bigr)^{\frac{1}{2}}.
\end{align*}
For a nice introduction to the spaces $W^{1,p}(\mathcal{I}_T,\mathcal{X})$ we refer to Section 1.4 in \cite{Barbu}. Let us just mention one result that we will often use: For $\psi\in L^p(\mathcal{I}_T;\mathcal{X})$ the condition that $\psi\in W^{1,p}(\mathcal{I}_T;\mathcal{X})$ is equivalent to the existence of an absolutely continuous $\psi_0:\mathcal{I}_T\to\mathcal{X}$ with strong derivative $\partial_t\psi_0: t\mapsto\lim_{h\to 0}\frac{\psi_0(t+h)-\psi_0(t)}{h}$ in $L^p(\mathcal{I}_T;\mathcal{X})$ such that $\psi(t)=\psi_0(t)$ for almost all $t\in\mathcal{I}_T$. Moreover, the Sobolev embedding $W^{1,p}(\mathcal{I}_T;\mathcal{X})\hookrightarrow_{p,T} L^\infty(\mathcal{I}_T;\mathcal{X})$ holds true. If $(\mathcal{Y},\|\cdot\|_{\mathcal{Y}})$ is another Banach space we let $(\mathcal{L}(\mathcal{X},\mathcal{Y}),\|\cdot\|_{\mathcal{L}(\mathcal{X},\mathcal{Y})})$ denote the Banach space of bounded linear operators $\mathcal{X}\to\mathcal{Y}$ and set $\mathcal{L}(\mathcal{X})=\mathcal{L}(\mathcal{X},\mathcal{X})$. By $A\lesssim B$ we mean that there exists a universal constant $c>0$ such that $A\leq cB$. Finally, we let $p'=\frac{p}{p-1}$ denote the Hölder conjugate to a given $p\in[1,\infty]$ and set $\langle s\rangle = \sqrt{1+s^2}$ for $s\in \mathbb{R}$.

\section*{Acknowledgements}
I would like to thank my advisor Professor Jan Philip Solovej for many helpful discussions.

\chapter{Motivation for the Model}\label{Motivation}
As our starting point we use the Abraham model of charged particles. So for some arbitrary $R>0$ and some positive $C_0^\infty$-function $\chi$ with $\int_{\mathbb{R}^3}\chi(\bm{x})\,\mathrm{d}\bm{x}=1$ we set $\chi_R:\bm{x}\mapsto \frac{1}{R^3}\chi\bigl(\frac{\bm{x}}{R}\bigr)$ and associate the smeared out charge distribution $\rho_{R,j}:\bm{x}\mapsto Q_j\chi_R(\bm{x}_j-\bm{x})$ to the $j$'th particle -- the corresponding Maxwell equations can be written as
\begin{align}
\mathrm{div}\bm{B}(t)&=0,\label{Maxwell1}\\
\nabla\times \bm{E}(t)&=-\frac{1}{c}\partial_t\bm{B}(t),\label{Maxwell2}\\
\mathrm{div}\bm{E}(t)&=4\pi\sum_{j=1}^N \rho_{R,j}(t),\label{Maxwell3}\\
\nabla\times \bm{B}(t)&=\frac{1}{c}\Bigl(\partial_t\bm{E}(t)+4\pi\sum_{j=1}^N \frac{\mathrm{d}\bm{x}_j}{\mathrm{d}t}(t)\rho_{R,j}(t)\Bigr),\label{Maxwell4}
\end{align}
and the Lorentz force law states that
\begin{align}
m_j\frac{\mathrm{d}^2\bm{x}_j}{\mathrm{d}t^2}(t)=Q_j\Bigl(\frac{1}{c}\frac{\mathrm{d}\bm{x}_j}{\mathrm{d}t}(t)\times \bm{B}(t)+\bm{E}(t)\Bigr)*\chi_R(\bm{x}_j(t))\quad\textrm{for }j\in\{1,\ldots,N\},\label{Lorentz}
\end{align}
where we interpret the coordinates of $\bm{x}=(\bm{x}_1,\ldots,\bm{x}_N)\in\mathbb{R}^{3N}$ as the positions of the $N$ particles, $\bm{B}$ is the magnetic field and $\bm{E}$ denotes the electric field. The reason for smearing out the charges is that the coupled Maxwell-Lorentz system does not make sense in the point particle case as explained in \cite{Spohn}. 

Now, \eqref{Maxwell1} ensures that $\bm{B}(t):\mathbb{R}^3\to\mathbb{R}^3$ can be written as the curl of some magnetic vector potential $\bm{A}(t):\mathbb{R}^3\to\mathbb{R}^3$, whereby \eqref{Maxwell2} allows us to write $-\bm{E}(t)-\frac{1}{c}\partial_t\bm{A}(t):\mathbb{R}^3\to\mathbb{R}^3$ as the gradient of some electric scalar potential $V:\mathbb{R}^3\to\mathbb{R}$. In other words,
\begin{align}
\bm{B}(t)=\nabla\times \bm{A}(t)\textrm{ and }\bm{E}(t)=-\frac{1}{c}\partial_t \bm{A}(t)-\nabla V(t).\label{potentials}
\end{align}
The choice of potentials is not unique -- if $(V,\bm{A})$ is an electromagnetic potential corresponding to the fields $\bm{E}$ and $\bm{B}$ then for any $\eta(t):\mathbb{R}^3\to\mathbb{R}$ the pair $\bigl(V-\frac{1}{c}\partial_t\eta,\bm{A}+\nabla \eta\bigr)$ will also serve as such a potential. This freedom of choice allows us to demand that $\bm{A}$ satisfies the Coulomb gauge condition \eqref{Coulombgauge}.

To formulate the problem in the Lagrangian formalism we choose the Hilbert manifold $\mathcal{Q}^0=\mathbb{R}^{3N} \times D^1 \times PL^2$ as configuration space, where $D^1$ is the space of locally integrable mappings that vanish at infinity and have square integrable first derivatives. Then the formulas \eqref{Maxwell3}--\eqref{Lorentz} are the Euler-Lagrange equations associated with the Lagrangian 
\begin{align*}
\mathscr{L}_R\bigl(\bm{x},V,\bm{A},\dot{\bm{x}},\dot{V},\dot{\bm{A}}\bigr)&=\sum_{j=1}^N \Bigl(\frac{1}{2}m_j\dot{\bm{x}}_j^2+\frac{Q_j}{c}\dot{\bm{x}}_j\cdot\bm{A}*\chi_R(\bm{x}_j)-Q_j V*\chi_R(\bm{x}_j)\Bigr)\\
&+\frac{1}{8\pi}\int_{\mathbb{R}^3}\Bigl(\Bigl|\frac{1}{c}\dot{\bm{A}}(\bm{y})+\nabla V(\bm{y})\Bigr|^2-|\nabla\times\bm{A}(\bm{y})|^2\Bigr)\,\mathrm{d}\bm{y}.
\end{align*}
defined on the restricted tangent bundle $T\mathcal{Q}^0|\mathcal{Q}^1\cong \mathcal{Q}^1\times \mathcal{Q}^0$, where $\mathcal{Q}^1$ denotes the manifold domain $\mathbb{R}^{3N}\times D^1\times PH^1$ of $\mathcal{Q}^0$. The associated energy function is $\mathscr{E}_R:T\mathcal{Q}^0|\mathcal{Q}^1\ni v\mapsto \bigl(\mathbb{F}\mathscr{L}_R(v)(v)-\mathscr{L}_R(v)\bigr)\in\mathbb{R}$, where the fiber derivative $\mathbb{F}\mathscr{L}_R: T\mathcal{Q}^0|\mathcal{Q}^1\to T^*\mathcal{Q}^0|\mathcal{Q}^1$ is given by
\begin{align*}
\mathbb{F}\mathscr{L}_R(v)(w)=\frac{\mathrm{d}}{\mathrm{d}t}\mathscr{L}_R(v+tw)\Big|_{t=0}\textrm{ for }q\in \mathcal{Q}^1\textrm{ and }v,w\in T_q\mathcal{Q}^0.
\end{align*}
With the intention of later passing to a quantum mechanical description of the charged particles we would like to define a Hamiltonian corresponding to $\mathscr{L}_R$ -- such a Hamiltonian expresses the energy in terms of coordinates and momenta, in the sense that the identity
\begin{align}
\mathscr{H}_R\circ \mathbb{F}\mathscr{L}_R=\mathscr{E}_R\label{Hamiltondef}
\end{align}
holds on some appropriate subset of $T\mathcal{Q}^0|\mathcal{Q}^1$ as we shall explain. The Lagrangian $\mathscr{L}_R$ is degenerate since it does not at all depend on $\dot{V}$ and so $\mathbb{F}\mathscr{L}_R$ is not even locally invertible, but as can easily be verified \eqref{Hamiltondef} does define a mapping $\mathscr{H}_R$ on all of the image $\mathcal{M}_1=\mathbb{F}\mathscr{L}_R(T\mathcal{Q}^0|\mathcal{Q}^1)\subset T^*\mathcal{Q}^0$. The pull-back $\omega_1=j_1^*\Omega$ to $\mathcal{M}_1$ of the canonical $2$-form $\Omega$ on $T^*\mathcal{Q}^0$ via the inclusion $\mathcal{M}_1\stackrel{j_1}{\longhookrightarrow} T^*\mathcal{Q}^0$ is degenerate and so $(\mathcal{M}_1,\omega_1)$ is not a symplectic manifold. To remedy this problem we can restrict $\mathbb{F}\mathscr{L}_R$ to the subset of elements $(\bm{x},V,\bm{A},\dot{\bm{x}},\dot{V},\dot{\bm{A}})\in T\mathcal{Q}^0|\mathcal{Q}^1$ satisfying Gauss' law
\begin{align*}
-\Delta V(\bm{z})=4\pi\sum_{j=1}^N Q_j\chi_R(\bm{x}_j-\bm{z}),
\end{align*}
meaning that $V$ is the function $\bm{z}\mapsto\sum_{j=1}^N Q_j\int_{\mathbb{R}^3}\frac{\chi_R(\bm{x}_j-\bm{y})}{|\bm{y}-\bm{z}|}\,\mathrm{d}\bm{y}$. The image $\mathcal{M}_2$ of this set under the map $\mathbb{F}\mathscr{L}_R$ becomes a weak symplectic manifold in the sense of \cite{Marsden} and this procedure is completely natural in the framework devised by Gotay, Nester and Hinds \cite{GotayNesterHinds} as a further development of Anderson, Bergmann and Dirac's constraint theory \cite{Bergmann, Dirac1, Dirac2} -- see also \cite{KP3}. Identifying $\mathcal{M}_2$ with $\mathbb{R}^{3N}\times PH^1\times \mathbb{R}^{3N}\times PL^2$ we can write the Hamiltonian $\mathscr{H}_R$ as
\begin{align}
&\mathscr{H}_R\bigl(\bm{x},\bm{A},\bm{p},-\tfrac{P\bm{E}}{4\pi}\bigr)\nonumber\\
&=\sum_{j=1}^N \frac{1}{2m_j}\Bigl(\bm{p}_j-\frac{Q_j}{c} \bm{A}*\chi_R(\bm{x}_j)\Bigr)^2\!+\frac{1}{8\pi}\!\int_{\mathbb{R}^{3N}}\!\bigl(c^2|P\bm{E}(\bm{y})|^2+|\nabla\times \bm{A}(\bm{y})|^2\bigr)\,\mathrm{d}\bm{y}\nonumber\\
&+\frac{1}{2}\sum_{j=1}^N\sum_{k=1}^N Q_jQ_k\int_{\mathbb{R}^3}\int_{\mathbb{R}^3}\frac{\chi_R(\bm{x}_j-\bm{y})\chi_R(\bm{x}_k-\bm{z})}{|\bm{y}-\bm{z}|}\,\mathrm{d}\bm{y}\,\mathrm{d}\bm{z}.\label{Hamilton}
\end{align}
Now take the point particle-limit $R\to 0^+$ in the following (formal) sense: Consider the mapping $\mathscr{H}_R$ acting as prescribed in \eqref{Hamilton} on the $R$-independent space $\mathbb{R}^{3N}\times PH^1\times \mathbb{R}^{3N}\times PL^2$. The first term on the right hand side of \eqref{Hamilton} represents the kinetic energy of the $N$ particles, the second term is the energy stored in the electromagnetic field and the double sum is the potential energy induced by the Coulomb interactions between the $N$ particles. In particular, the double sum's diagonal term $\frac{Q_j^2}{2R}\int_{\mathbb{R}^3}\int_{\mathbb{R}^3}\frac{\chi(\bm{y})\chi(\bm{z})}{|\bm{y}-\bm{z}|}\,\mathrm{d}\bm{y}\,\mathrm{d}\bm{z}$ is the energy coming from the $j$'th particle's interaction with itself. We subtract this self-energy from $\mathscr{H}_R$ and note that as $R\to 0^+$ the result converges pointwise to the mapping
\begin{align*}
\mathscr{H}_0\bigl(\bm{x},\bm{A},\bm{p},-\tfrac{P\bm{E}}{4\pi}\bigr)&=\sum_{j=1}^N \frac{1}{2m_j}\Bigl(\bm{p}_j-\frac{Q_j}{c} \bm{A}(\bm{x}_j)\Bigr)^2+\sum_{1\leq j<k\leq N} \frac{Q_jQ_k}{|\bm{x}_j-\bm{x}_k|}\nonumber\\
&+\frac{1}{8\pi}\int_{\mathbb{R}^{3N}}\bigl(c^2|P\bm{E}(\bm{y})|^2+|\nabla\times \bm{A}(\bm{y})|^2\bigr)\,\mathrm{d}\bm{y},
\end{align*}
provided $\bm{A}$ is continuous at the points $\bm{x}_1,\ldots,\bm{x}_N$. We now quantize the charged particles in our model and obtain the Hamilton operator
\begin{align*}
\mathscr{H}\bigl(\bm{A},-\tfrac{P\bm{E}}{4\pi}\bigr)&=\sum_{j=1}^N \frac{1}{2m_j}\Bigl(i\hbar \nabla_{\bm{x}_j}+\frac{Q_j}{c} \bm{A}(\bm{x}_j)\Bigr)^2+\sum_{1\leq j<k\leq N} \frac{Q_jQ_k}{|\bm{x}_j-\bm{x}_k|}\nonumber\\
&+\frac{1}{8\pi}\int_{\mathbb{R}^{3N}}\bigl(c^2|P\bm{E}(\bm{y})|^2+|\nabla\times \bm{A}(\bm{y})|^2\bigr)\,\mathrm{d}\bm{y},
\end{align*}
acting on a certain dense subspace of the Hilbert space $L^2(\mathbb{R}^{3N})$. Instead of also quantizing the fields $\bm{A}$ and $-\frac{P\bm{E}}{4\pi}$ we leave them as classical variables. In this spirit we will for a given (normalized) quantum state $\psi:\mathbb{R}^{3N}\to \mathbb{C}$ of the particles regard the average energy $(\bm{A},-\frac{P\bm{E}}{4\pi})\mapsto \bigl(\psi,\mathscr{H}\bigl(\bm{A},-\tfrac{P\bm{E}}{4\pi}\bigr)\psi\bigr)_{L^2}$ as a classical Hamiltonian defined on the weak symplectic manifold $(PH^1\times PL^2,\omega)$ with
\begin{align*}
\omega_m\bigl(m,\bm{A}_1,-\tfrac{P\bm{E}_1}{4\pi},m,\bm{A}_2,-\tfrac{P\bm{E}_2}{4\pi}\bigr)=\frac{1}{4\pi}\int_{\mathbb{R}^3}\bigl(P\bm{E}_1\cdot\bm{A}_2-P\bm{E}_2\cdot\bm{A}_1\bigr)(\bm{y})\,\mathrm{d}\bm{y}
\end{align*}
for $m,\bigl(\bm{A}_1,-\frac{P\bm{E}_1}{4\pi}\bigr),\bigl(\bm{A}_2,-\frac{P\bm{E}_2}{4\pi}\bigr)\in PH^1\times PL^2$. The corresponding Hamilton equations express that
\begin{align}
\frac{1}{c^2}\partial_t\bm{A}(t)=-P\bm{E}(t)\textrm{ and }-\partial_t P\bm{E}(t)=\Delta\bm{A}(t)+\frac{4\pi}{c}\sum_{j=1}^NP\bm{J}_j[\psi,\bm{A}(t)].\label{HamiltQ}
\end{align}
In reality, we do of course not expect the quantum state of the charged particles to be time independent -- the time evolution of $\psi$ is governed by the Schrödinger equation
\begin{align}
i\hbar \partial_t\psi(t)=\mathscr{H}\bigl(\bm{A}(t),-\tfrac{P\bm{E}(t)}{4\pi}\bigr)\psi(t).\label{SchrodQ}
\end{align}
We investigate the situation where the fixed time-independent state $\psi$ appearing in \eqref{HamiltQ} is replaced by the time-dependent state $\psi(t)$ satisfying the Schrödinger equation \eqref{SchrodQ}. \eqref{MSP} is precisely obtained by doing this coupling of \eqref{HamiltQ} with \eqref{SchrodQ}. 

\chapter{Preliminaries}\label{proof}
First, we collect some simple estimates that will be useful to us later.
\begin{lemma}\label{Estlemma}
For all $1\leq j\leq N$, $\bm{A}\in \bigl[L^4(\mathbb{R}^3)\bigr]^3$, $B\in L^2(\mathbb{R}^3)$, $\psi\in L^2(\mathbb{R}^{3N})$ with $\Delta_{\bm{x}_j}\psi\in L^2(\mathbb{R}^{3N})$, $0<\varepsilon<1$ and $0<\delta<\frac{1}{2}$ we have
\begin{align}
\begin{split}\label{KatoRellich}
\bigl\|\bm{A}(\bm{x}_j)\cdot\nabla_{\bm{x}_j}\psi\bigr\|_{L^2}
&\lesssim \|\bm{A}\|_{L^4}\bigl(\varepsilon^{-7}\|\psi\|_{L^2}+\varepsilon\|\Delta_{\bm{x}_j}\psi\|_{L^2}\bigr)\\
\|B(\bm{x}_j)\psi\|_{L^2}&\lesssim \|B\|_{L^2}\bigl(\varepsilon^{-\frac{3+2\delta}{1-2\delta}}\|\psi\|_{L^2}+\varepsilon\|\Delta_{\bm{x}_j}\psi\|_{L^2}\bigr)\\
\left\|\frac{1}{|\bm{x}_j-\bm{x}_k|}\psi\right\|_{L^2}&\lesssim \varepsilon^{-\frac{3+2\delta}{1-2\delta}}\|\psi\|_{L^2}+\varepsilon\|\Delta_{\bm{x}_j}\psi\|_{L^2}\textrm{ for }1\leq j<k\leq N
\end{split}
\end{align}
and for all $1\leq j\leq N$, $\bm{A}\in \bigl[L^4(\mathbb{R}^3)\bigr]^3$, $B\in L^2(\mathbb{R}^3)$ and $\psi\in L^2(\mathbb{R}^{3N})$ the estimates
\begin{align}
\begin{split}\label{Hminzwei}
\|\mathrm{div}_{\bm{x}_j}(\bm{A}(\bm{x}_j)\psi)\|_{H^{-2}}&\lesssim \|\bm{A}\|_{L^4}\|\psi\|_{L^2},\\
\|B(\bm{x}_j)\psi\|_{H^{-2}}&\lesssim \|B\|_{L^2}\|\psi\|_{L^2},\\
\left\|\frac{1}{|\bm{x}_j-\bm{x}_k|}\psi\right\|_{H^{-2}}&\lesssim \|\psi\|_{L^2}\textrm{ for }1\leq j<k\leq N
\end{split}
\end{align}
hold true. Moreover, we have
\begin{align}
&\bigl\|\bm{J}_j[\psi_1,\bm{A}_1]-\bm{J}_j[\psi_2,\bm{A}_2]\bigr\|_{H^1}\nonumber\\
&\lesssim\sum_{k=1}^2\bigl((1+\|\bm{A}_k\|_{D^1})\|\psi_k\|_{H^2}\bigr)\|\psi_1-\psi_2\|_{H^2}+\|\psi_1\|_{H^2}\|\psi_2\|_{H^2}\|\bm{A}_1-\bm{A}_2\|_{D^1}\label{current}
\end{align}
for any $1\leq j\leq N$ and $(\psi_1,\bm{A}_1),(\psi_2,\bm{A}_2)\in H^2(\mathbb{R}^{3N})\times D^1(\mathbb{R}^3)$.
\end{lemma}
\begin{proof}
For instance we can use Tonelli's theorem, Hölder's inequality, the Sobolev embedding $H^{\frac{3}{4}}\hookrightarrow L^4$ as well as the Young inequalities $\bm{p}_j^2\leq \frac{1}{2\varepsilon^2}+\frac{\varepsilon^2}{2}\bm{p}_j^4$ and $|\bm{p}_j|^{\frac{7}{2}}\leq \frac{1}{8\varepsilon^{14}}+\frac{7\varepsilon^2}{8}\bm{p}_j^4$ to obtain
\begin{align*}
\|\bm{A}(\bm{x}_j)\cdot\nabla_{\bm{x}_j}\psi\|_{L^2}^2&\leq\int_{\mathbb{R}^{3(N-1)}}\Bigl(\int_{\mathbb{R}^3}|\bm{A}(\bm{x}_j)|^4\,\mathrm{d}\bm{x}_j\Bigr)^{\frac{1}{2}}\Bigl(\int_{\mathbb{R}^3}| \nabla_{\bm{x}_j}\psi(\bm{x})|^4\,\mathrm{d}\bm{x}_j\Bigr)^{\frac{1}{2}}\,\mathrm{d}\bm{x}_j'\\
&\lesssim \|\bm{A}\|_{L^4}^2\int_{\mathbb{R}^{3(N-1)}}\int_{\mathbb{R}^3}\bigl|(1-\Delta)^{\frac{3}{8}}\nabla\psi^{j,\bm{x}_j'}(\bm{x}_j)\bigr|^2\,\mathrm{d}\bm{x}_j\,\mathrm{d}\bm{x}_j'\\
&\lesssim \|\bm{A}\|_{L^4}^2\bigl(\varepsilon^{-14}\|\psi\|_{L^2}^2+\varepsilon^2\|\Delta_{\bm{x}_j}\psi\|_{L^2}^2\bigr),
\end{align*}
where we for $\bm{x}_j'=(\bm{x}_1,\ldots,\bm{x}_{j-1},\bm{x}_{j+1},\ldots,\bm{x}_N)\in\mathbb{R}^{3(N-1)}$ introduce the mapping $\psi^{j,\bm{x}_j'}:\bm{x}_j\mapsto \psi(\bm{x}_1,\ldots,\bm{x}_{j-1},\bm{x}_j,\bm{x}_{j+1},\ldots,\bm{x}_N)$ that for almost all vectors $\bm{x}_j'$ is contained in the Sobolev space $H^2(\mathbb{R}^3)$ and satisfies the identities $(1-\Delta)^{\frac{s}{2}}\bigl[\psi^{j,\bm{x}_j'}\bigr]=\bigl[(1-\Delta_{\bm{x}_j})^{\frac{s}{2}}\psi\bigr]^{j,\bm{x}_j'}$ and $\partial^\alpha\bigl[\psi^{j,\bm{x}_j'}\bigr]=[\partial_{\bm{x}_j}^\alpha\psi]^{j,\bm{x}_j'}$ for any $s\leq 2$ and any multi-index $\alpha$ with $|\alpha|\leq 2$. The other estimates in \eqref{KatoRellich} follow analogously by using the Sobolev embedding $H^{\frac{3}{2}+\delta}\hookrightarrow L^\infty$ instead of $H^{\frac{3}{4}}\hookrightarrow L^4$. 

To prove the first inequality in \eqref{Hminzwei} we first note that for $\xi\in C_0^\infty$,
\begin{align}
\|\mathrm{div}_{\bm{x}_j}(\bm{A}(\bm{x}_j)\xi)\|_{H^{-2}}&\leq (2\pi)^{3N}\sum_{k=1}^3\sup_{\|\eta\|_{L^2}=1}\left|\left(\overline{A^k}(\bm{x}_j)\mathscr{F}^{-1}\bigl[(1+\bm{p}^2)^{-\frac{1}{2}}\eta\bigr],\xi\right)_{L^2}\right|\nonumber\\
&\lesssim \|\bm{A}\|_{L^4}\|\xi\|_{L^2},\label{vurdC}
\end{align}
where we use the Riesz-Fréchet theorem and the Sobolev embedding $H^1\hookrightarrow L^4$. For a given $\psi\in L^2$ we can therefore choose a sequence $(\psi_n)_{n\in\mathbb{N}}$ of $C_0^\infty$-functions converging in $L^2$ to $\psi$ and use \eqref{vurdC} to conclude that $(\mathrm{div}_{\bm{x}_j}(\bm{A}(\bm{x}_j)\psi_n))_{n\in\mathbb{N}}$ is a Cauchy sequence in the Hilbert space $H^{-2}$, whereby it must converge to some limit in $H^{-2}$. But this limit has to be $\mathrm{div}_{\bm{x}_j}(\bm{A}(\bm{x}_j)\psi)$ since the convergence $\mathrm{div}_{\bm{x}_j}(\bm{A}(\bm{x}_j)\psi_n)\xrightarrow[n\to\infty]{} \mathrm{div}_{\bm{x}_j}(\bm{A}(\bm{x}_j)\psi)$ holds in the space $\mathscr{D}'$ of distributions. Thus, the first estimate of \eqref{Hminzwei} is true and each of the remaining two inequalities follow by combining the Riesz-Fréchet theorem with the corresponding estimate in \eqref{KatoRellich}.

Finally, \eqref{current} is easy to derive from the general estimates 
\begin{align*}
&\Bigl(\int_{\mathbb{R}^3}\Bigl|\int_{\mathbb{R}^{3(N-1)}}\Psi_1(\bm{x})\nabla_{\bm{x}_j}\Psi_2(\bm{x})\,\mathrm{d}\bm{x}_j'\Bigr|^2\,\mathrm{d}\bm{x}_j\Bigr)^{\frac{1}{2}}\\
&\lesssim \min\bigl\{\|(1-\Delta_{\bm{x}_j})^{\frac{1}{4}}\Psi_1\|_{L^2}\|\nabla_{\bm{x}_j}\otimes\nabla_{\bm{x}_j}\Psi_2\|_{L^2},\\
&\hspace{6.5cm}\|(1-\Delta_{\bm{x}_j})^{\frac{3}{4}+\frac{\delta}{2}}\Psi_1\|_{L^2}\|\nabla_{\bm{x}_j}\Psi_2\|_{L^2}\bigr\}
\end{align*}
and
\begin{align*}
&\Bigl(\int_{\mathbb{R}^3}\Bigl|\int_{\mathbb{R}^{3(N-1)}}\bm{A}(\bm{x}_j)\Psi_1(\bm{x})\Psi_2(\bm{x})\,\mathrm{d}\bm{x}_j'\Bigr|^2\,\mathrm{d}\bm{x}_j\Bigr)^{\frac{1}{2}}\\
&\lesssim\min\Bigl\{\|\bm{A}\|_{L^6}\|\nabla_{\bm{x}_j}\Psi_1\|_{L^2}\|\nabla_{\bm{x}_j}\Psi_2\|_{L^2},\|\bm{A}\|_{L^2}\prod_{k=1}^2\|(1-\Delta_{\bm{x}_j})^{\frac{3}{4}+\frac{\delta}{2}}\Psi_k\|_{L^2}\Bigr\}
\end{align*}
on mappings $\Psi_1,\Psi_2:\mathbb{R}^{3N}\to \mathbb{C}$ and $\bm{A}:\mathbb{R}^3\to \mathbb{C}^3$ that follow for $\delta>0$ from Minkowski's integral inequality, the Sobolev embeddings $D^1\hookrightarrow L^6$, $H^{\frac{1}{2}}\hookrightarrow L^3$, $H^{\frac{3}{2}+\delta}\hookrightarrow L^\infty$ and Hölder's inequality. By $\nabla_{\bm{x}_j}\otimes\nabla_{\bm{x}_j}\Psi_2$ we here mean a $9$-vector with the derivatives $\partial_{x_j^k}\partial_{x_j^\ell}\Psi_2$ as components ($k,\ell\in\{1,2,3\}$).
\end{proof}
\begin{bemaerkning}
The lemma above allows us to clarify the exact meaning of a solution to \eqref{MSP}. If for some given pair $(\psi,\bm{A})\in C(\mathcal{I}_{T},H^2)\times  C\bigl(\mathcal{I}_T;H^{\frac{3}{2}}\bigr)$ the derivative $\partial_t\bm{A}$ of $\bm{A}\in \mathscr{D}'\bigl(\mathcal{I}_T^\circ;H^{\frac{3}{2}}\bigr)$ is a continuous mapping $\mathcal{I}_{T}\to H^{\frac{1}{2}}$ then by boundedness of $P:H^1\to H^1$ and the estimates in Lemma \ref{Estlemma} we have
\begin{align}
c^2\Bigl(\Delta\bm{A}+\frac{4\pi}{c}\sum_{j=1}^N P\bm{J}_j[\psi,\bm{A}]\Bigr)\in C\bigl(\mathcal{I}_{T};H^{-\frac{1}{2}}\bigr)\label{Asol}
\end{align}
and
\begin{align}
-\frac{i}{\hbar}\Bigl(\sum_{j=1}^N \frac{1}{2m_j}\nabla_{j,\bm{A}}^2\psi+\sum_{1\leq j<k\leq N} \frac{Q_jQ_k}{|\bm{x}_j-\bm{x}_k|}\psi+\mathscr{E}_{\mathrm{EM}}[\bm{A},\partial_t\bm{A}]\psi\Bigr)\in C(\mathcal{I}_{T};L^2).\label{psisol}
\end{align}
A pair $(\psi,\bm{A})\in C(\mathcal{I}_T;H^2)\times \bigl(C\bigl(\mathcal{I}_T;H^{\frac{3}{2}}\bigr)\cap C^1\bigl(\mathcal{I}_T;H^{\frac{1}{2}}\bigr)\bigr)$ is said to solve \eqref{MSP} if the second derivative $\partial_t^2\bm{A}$ of $\bm{A}\in \mathscr{D}'\bigl(\mathcal{I}_T^\circ;H^{\frac{3}{2}}\bigr)$ equals \eqref{Asol} and the derivative $\partial_t\psi$ of $\psi\in \mathscr{D}'(\mathcal{I}_T^\circ;H^2)$ equals \eqref{psisol}.
\end{bemaerkning}
For any solution $(\psi,\bm{A})\in C(\mathcal{I}_T;H^2)\times \bigl(C\bigl(\mathcal{I}_T;H^{\frac{3}{2}}\bigr)\cap C^1\bigl(\mathcal{I}_T;H^{\frac{1}{2}}\bigr)\bigr)$ to \eqref{MSPmaerke} the pair $\bigl(\mathrm{e}^{-\frac{i}{\hbar}\int_0^t\mathscr{E}_{\mathrm{EM}}[\bm{A},\partial_t\bm{A}](s)\,\mathrm{d}s}\psi,\bm{A}\bigr)$ will solve \eqref{MSP} -- here, the field energy $\mathscr{E}_{\mathrm{EM}}[\bm{A},\partial_t\bm{A}]$ is absolutely continuous $\mathcal{I}_{T}\to\mathbb{R}$ because $\partial_t\bm{A}$ and $\nabla\times \bm{A}$ are both absolutely continuous $\mathcal{I}_{T}\to H^{-\frac{1}{2}}$ and continuous $\mathcal{I}_{T}\to H^{\frac{1}{2}}$. Conversely, any $(\psi,\bm{A})\in C(\mathcal{I}_T;H^2)\times \bigl(C\bigl(\mathcal{I}_T;H^{\frac{3}{2}}\bigr)\cap C^1\bigl(\mathcal{I}_T;H^{\frac{1}{2}}\bigr)\bigr)$ solving \eqref{MSP} gives rise to the solution $\bigl(\mathrm{e}^{\frac{i}{\hbar}\int_0^t\mathscr{E}_{\mathrm{EM}}[\bm{A},\partial_t\bm{A}](s)\,\mathrm{d}s}\psi,\bm{A}\bigr)$ to \eqref{MSPmaerke}. Therefore we can concentrate on uniquely solving the simplified initial value problem \eqref{MSPmaerke}+\eqref{initialcondi} instead of \eqref{MSP}+\eqref{initialcondi}.

It is noteworthy that for any solution $(\psi,\bm{A})$ to the system \eqref{MSPmaerke} (or \eqref{MSP} for that matter) the norm $\|\psi\|_{L^2}:\mathcal{I}_{T}\ni t\mapsto \|\psi(t)\|_{L^2}\in\mathbb{R}$ will be a constant of the motion. The absolute continuity of $\psi:\mathcal{I}_{T}\to L^2$ implies namely that $\|\psi\|_{L^2}^2$ is absolutely continuous and for almost all $t\in\mathcal{I}_{T}$
\begin{align}
\partial_t\|\psi\|_{L^2}^2(t)=\frac{2}{\hbar}\mathrm{Im}\left(\psi(t),\sum_{j=1}^N \frac{1}{2m_j}\nabla_{j,\bm{A}}^2\psi(t)+\sum_{1\leq j<k\leq N}\frac{Q_jQ_k}{|\bm{x}_j-\bm{x}_k|}\psi(t)
\right)_{L^2}\!\!=0.\label{consi}
\end{align}
So if the initial condition $\psi_0$ is a unit vector in $L^2$ then the wave function $\psi$ will continue to be a unit vector in $L^2$ at all later times of existence -- this is consistent with the quantum mechanical interpretation of $|\psi(t)(\bm{x}_1,\ldots,\bm{x}_N)|^2$ as the probability density at time $t$ for finding particle $1$ at $\bm{x}_1$, particle $2$ at $\bm{x}_2$ etc.

Let us emphasize a final important consequence of Lemma \ref{Estlemma} -- namely that for any choice of divergence free vector potential $\bm{A}\in L^4(\mathbb{R}^3;\mathbb{R}^3)$ the formal operator acting on $\psi$ on the right hand side of the second equation in \eqref{MSPmaerke} can be realized as a symmetric operator in $L^2(\mathbb{R}^{3N})$ with dense domain $H^2(\mathbb{R}^{3N})$. By the Kato-Rellich theorem the selfadjointness of the nonnegative operator $-\sum_{j=1}^N\frac{\hbar^2}{2m_j}\Delta_{\bm{x}_j}:H^2(\mathbb{R}^{3N})\to L^2(\mathbb{R}^{3N})$ and the estimates \eqref{KatoRellich} even imply that $\sum_{j=1}^N \frac{1}{2m_j}\nabla_{j,\bm{A}}^2+\sum_{1\leq j<k\leq N} \frac{Q_jQ_k}{|\bm{x}_j-\bm{x}_k|}$ is selfadjoint on the domain $H^2(\mathbb{R}^{3N})$ with a lower bound that goes like some power of $\langle\|\bm{A}\|_{L^4}\rangle$.

\chapter{The Many-Body Schrödinger Equation}\label{SchEq}
We will eventually solve \eqref{MSPmaerke} by applying the Banach fixed-point theorem to the solution operator of a certain linearization of \eqref{MSPmaerke}. In this section we approach the many-body Schrödinger equation
\begin{align}
i\hbar\partial_t \xi &=\Bigl(\sum_{j=1}^N \frac{1}{2m_j}\nabla_{j,\bm{A}}^2+\sum_{1\leq j<k\leq N}\frac{Q_jQ_k}{|\bm{x}_j-\bm{x}_k|}\Bigr)\xi\label{LinSchr}
\end{align}
by considering $\bm{A}$ as a fixed (time-dependent) vector potential. We supply \eqref{LinSchr} with the initial condition
\begin{align}
\xi(\tau)&= \psi_0,\label{LinSchrI}
\end{align}
where $\tau\in\mathcal{I}_{T}$ and $\psi_0$ are also fixed and thought of as given beforehand. We will show that this initial value problem is well-posed by applying the following fundamental result by Kato concerning general linear evolution equations of the type
\begin{align*}
\partial_t\xi+\mathbb{A}(t)\xi&=\mathbb{F}(t),\\
\xi(\tau)=\psi_0
\end{align*}
in a Banach space $\mathcal{X}$.
\begin{saetning}{\emph{\cite[Theorem I]{Kato}}}\label{Katoth}
Suppose that
\begin{enumerate}
\item[\emph{(i')}] For all $t\in\mathcal{I}_T$ the operator $-\mathbb{A}(t)$ generates a strongly continuous one-parameter semigroup $[0,\infty)\ni s\mapsto \exp\bigl(-s\mathbb{A}(t)\bigr)\in \mathcal{L}(\mathcal{X})$ and the family $\{\mathbb{A}(t)\,|\,t\in\mathcal{I}_T\}$ is quasi-stable with stability index $(M,\beta)$, in the sense that
\begin{align*}
\Bigl\|\prod_{j=1}^k \exp(-s_j\mathbb{A}(t_j))\Bigr\|_{\mathcal{L}(\mathcal{X})}\leq M\exp\Bigl(\sum_{j=1}^k s_j\beta(t_j)\Bigr)
\end{align*}
for all $k\in\mathbb{N}$, $0\leq t_1\leq\cdots\leq t_k\leq T$ and $s_1,\ldots,s_k\in[0,\infty)$, where $M$ is a constant, $\beta:\mathcal{I}_T\to \mathbb{R}$ is upper Lebesgue integrable and the product on the left hand side is time-ordered so that a factor with larger $t_j$ stands to the left of ones with smaller $t_j$.
\item[\emph{(ii'\,\!'\,\!')}] There exists a Banach space $\mathcal{Y}$, continuously and densely embedded in $\mathcal{X}$, and a family $\{\mathbb{S}(t)\,|\,t\in\mathcal{I}_T\}$ of isomorphisms $\mathcal{Y}\to\mathcal{X}$, such that
\begin{align*}
\mathbb{S}(t)\mathbb{A}(t)\mathbb{S}(t)^{-1}=\mathbb{A}(t)+\mathbb{B}(t)\textrm{ for almost all }t\in\mathcal{I}_T,
\end{align*}
where $\mathbb{B}$ maps into $\mathcal{L}(\mathcal{X})$, $\mathbb{B}(\cdot)x$ is strongly measurable \emph{(}as an $\mathcal{X}$-valued mapping\emph{)} for all $x\in\mathcal{X}$ and $\|\mathbb{B}(\cdot)\|_{\mathcal{L}(\mathcal{X})}$ is upper Lebesgue integrable. Furthermore, there exists a function $\dot{\mathbb{S}}$ defined almost everywhere on $\mathcal{I}_T$ and mapping into $\mathcal{L}(\mathcal{Y},\mathcal{X})$ such that $\dot{\mathbb{S}}(\cdot)y$ is strongly measurable for all $y\in\mathcal{Y}$, $\bigl\|\dot{\mathbb{S}}(\cdot)\bigr\|_{\mathcal{L}(\mathcal{Y},\mathcal{X})}$ is upper Lebesgue integrable and $\mathbb{S}$ is a strong indefinite integral of $\dot{\mathbb{S}}$.
\item[\emph{(iii)}] For all $t\in\mathcal{I}_T$ the domain of the operator $\mathbb{A}(t)$ in $\mathcal{X}$ contains $\mathcal{Y}$ and $\mathbb{A}:\mathcal{I}_T\to \mathcal{L}(\mathcal{Y},\mathcal{X})$ is norm-continuous.
\end{enumerate}
Then there exists a unique $\mathscr{U}$ defined on the triangle $\mathcal{T}_T=\{(t,\tau)\in \mathcal{I}_T^2\mid t\geq \tau\}$ with the following properties.
\begin{enumerate}
\item[\emph{(a)}] $\mathscr{U}$ is strongly continuous $\mathcal{T}_T\to \mathcal{L}(\mathcal{X})$ with $\mathscr{U}(t,t)=1$ for all $t\in\mathcal{I}_T$,
\item[\emph{(b)}] $\mathscr{U}(t,\tau)\mathscr{U}(\tau,s)=\mathscr{U}(t,s)$ for all $(t,\tau,s)$ satisfying $0\leq s\leq \tau\leq t\leq T$,
\item[\emph{(c)}] For all $(t,\tau)\in \mathcal{T}_T$ the inclusion $\mathscr{U}(t,\tau)\mathcal{Y}\subset \mathcal{Y}$ holds and $\mathscr{U}$ is strongly continuous $\mathcal{T}_T\to \mathcal{L}(\mathcal{Y})$,
\item[\emph{(d)}] The strong partial derivatives $\partial_t\mathscr{U}(t,\tau)y=-\mathbb{A}(t)\mathscr{U}(t,\tau)y$ as well as  $\partial_{\tau}\mathscr{U}(t,\tau)y=\mathscr{U}(t,\tau)\mathbb{A}(\tau)y$ exist in $\mathcal{X}$ for all $(t,\tau,y)\in \mathcal{T}_T\times \mathcal{Y}$ and $\partial_t\mathscr{U},\partial_{\tau}\mathscr{U}:\mathcal{T}_T\to \mathcal{L}(\mathcal{Y},\mathcal{X})$ are both strongly continuous.
\end{enumerate}
\end{saetning}
\begin{bemaerkning}\label{Katore}
If $\mathbb{A}$ satisfies the points (i'), (ii'\,\!'\,\!') and (iii) then $\mathbb{A}'=-\mathbb{A}\circ \Re$ with $\Re:\mathcal{I}_T\ni t\mapsto (T-t)\in \mathcal{I}_T$ will automatically fulfill (ii'\,\!'\,\!') and (iii). This can easily be checked by choosing $\bigl(\mathbb{S}',\mathbb{B}',\dot{\mathbb{S}}'\bigr)=\bigl(\mathbb{S},-\mathbb{B},-\dot{\mathbb{S}}\bigr)\circ \Re$ (with a hopefully obvious notation) and using that for any Banach space $\mathcal{Z}$ the function $f\mapsto (-f\circ \Re)$ not only conserves the property of strong measurability $\mathcal{I}_T\to \mathcal{Z}$, but it also maps $L^1(\mathcal{I}_T;\mathcal{Z})$ isometrically onto itself. If $\mathbb{A}'$ also happens to satisfy (i') in the sense that $-\mathbb{A}'(t)$ generates a $C_0$-semigroup for all $t\in\mathcal{I}_T$ and the family $\{\mathbb{A}'(t)\mid t\in\mathcal{I}_T\}$ is quasi-stable with stability index $(M,\beta\circ\Re)$, then we can combine the evolution operators $\mathscr{U}_{\mathbb{A}}$ and $\mathscr{U}_{\mathbb{A}'}$ -- whose existence are ensured by Theorem \ref{Katoth} -- into a single evolution operator $\mathscr{U}$ defined in \emph{all} points $(t,\tau)\in \mathcal{I}_T^2$ by setting
\begin{align*}
\mathscr{U}(t,\tau)=\begin{dcases}\mathscr{U}_{\mathbb{A}}(t,\tau)&\textrm{for }t\geq \tau\\\mathscr{U}_{\mathbb{A}'}(T-t,T-\tau)&\textrm{for }t< \tau\end{dcases}.
\end{align*}
This operator satisfies
\begin{enumerate}
\item[(a')] $\mathscr{U}$ is strongly continuous $\mathcal{I}_T^2\to \mathcal{L}(\mathcal{X})$ with $\mathscr{U}(t,t)=1$ for all $t\in\mathcal{I}_T$,
\item[(b')] $\mathscr{U}(t,\tau)\mathscr{U}(\tau,s)=\mathscr{U}(t,s)$ for all $(t,\tau,s)\in \mathcal{I}_T^3$,
\item[(c')] For all $(t,\tau)\in \mathcal{I}_T^2$ the inclusion $\mathscr{U}(t,\tau)\mathcal{Y}\subset \mathcal{Y}$ holds and $\mathscr{U}$ is strongly continuous $\mathcal{I}_T^2\to \mathcal{L}(\mathcal{Y})$,
\item[(d')] The strong partial derivatives $\partial_t\mathscr{U}(t,\tau)y=-\mathbb{A}(t)\mathscr{U}(t,\tau)y$ as well as  $\partial_{\tau}\mathscr{U}(t,\tau)y=\mathscr{U}(t,\tau)\mathbb{A}(\tau)y$ exist in $\mathcal{X}$ for all $(t,\tau,y)\in \mathcal{I}_T^2\times \mathcal{Y}$ and $\partial_t\mathscr{U},\partial_{\tau}\mathscr{U}:\mathcal{I}_T^2\to \mathcal{L}(\mathcal{Y},\mathcal{X})$ are both strongly continuous.
\end{enumerate}
Here, (b') is the only point that does not follow immediately from the properties listed in Theorem \ref{Katoth} of the individual operators $\mathscr{U}_{\mathbb{A}}$ and $\mathscr{U}_{\mathbb{A}'}$ -- however, it suffices to prove the identities
\begin{align}
\mathscr{U}_{\mathbb{A}}(t_0,\tau_0)\mathscr{U}_{\mathbb{A}'}(T-\tau_0,T-t_0)=\mathscr{U}_{\mathbb{A}'}(T-\tau_0,T-t_0)\mathscr{U}_{\mathbb{A}}(t_0,\tau_0)=1\label{bpoint}
\end{align}
for all $(t_0,\tau_0)\in\mathcal{T}_T$. To prove \eqref{bpoint} note first that by \cite[Proposition 4.4]{Kato1} the operator $\widetilde{\mathbb{A}}(t)$ (resp. $\widetilde{\mathbb{A}}'(t)$) in $\mathcal{Y}$ acting like $\mathbb{A}(t)$ (resp. $\mathbb{A}'(t)$) on the domain $\{y\in\mathcal{Y}\mid \mathbb{A}(t)y\in\mathcal{Y}\}$ (resp. $\{y\in\mathcal{Y}\mid \mathbb{A}'(t)y\in\mathcal{Y}\}$) is quasi-stable and the second coordinate of it's stability index can be chosen to be $\widetilde{\beta}=\beta+M\|\mathbb{B}(\cdot)\|_{\mathcal{L}(\mathcal{X})}$ (resp. $\widetilde{\beta}\circ\Re$). Without loss of generality we can here assume that $\beta$ and $\widetilde{\beta}$ are integrable $\mathcal{I}_T\to [0,\infty)$ (otherwise replace them by integrable majorants). With the help of \cite[Lemma A1]{Kato} and the remark after \cite[Lemma A2]{Kato} consider now a sequence $\bigl(\{\mathcal{I}_{T}^{n1},\ldots,\mathcal{I}_{T}^{nm_n}\}\bigr)_{n\in\mathbb{N}}$ of partitions of the interval $\mathcal{I}_{T}$ into subintervals with $\sup_j \bigl|\mathcal{I}_{T}^{nj}\bigr|\xrightarrow[n\to\infty]{} 0$ and a corresponding sequence $\bigl(\{t^{n1},\ldots,t^{n m_n}\}\bigr)_{n\in\mathbb{N}}$ with $t^{nj}\in \mathcal{I}_{T}^{nj}$ for $n\in\mathbb{N}$ and $j\in\{1,\ldots,m_n\}$ such that the Riemann step functions $\sum_{j=1}^{m_n}\beta\bigl(t^{nj}\bigr)1_{\mathcal{I}_T^{nj}}$ and $\sum_{j=1}^{m_n}\widetilde{\beta}\bigl(t^{nj}\bigr)1_{\mathcal{I}_T^{nj}}$ approximate $\beta$ respectively $\widetilde{\beta}$, in $L^1(\mathcal{I}_T)$ as well as pointwise almost everywhere. Then by the proof of \cite[Theorem I]{Kato} the operator $\mathscr{U}_{\mathbb{A}}(t,\tau)$ is the strong limit in $\mathcal{L}(L^2)$ (uniformly in $(t,\tau)\in \mathcal{T}_T$) of a sequence $\bigl(\mathscr{U}_{\mathbb{A}}^n(t,\tau)\bigr)_{n\in\mathbb{N}}$ of operators satisfying
\begin{itemize}
\item $\mathscr{U}_{\mathbb{A}}^n(t,\tau)=\mathrm{e}^{-(t-\tau)\mathbb{A}(t^{nj})}$ for $t,\tau\in \overline{\mathcal{I}^{nj}_{T}}$ with $t\geq \tau$,
\item $\mathscr{U}_{\mathbb{A}}^n(t,\tau)=\mathscr{U}_{\mathbb{A}}^n(t,s)\mathscr{U}_{\mathbb{A}}^n(s,\tau)$ for $t\geq s\geq \tau$.
\end{itemize}
But here the sequence $\bigl(\{T-\mathcal{I}_{T}^{n1},\ldots,T-\mathcal{I}_{T}^{nm_n}\}\bigr)_{n\in\mathbb{N}}$ of partitions of $\mathcal{I}_T$ satisfies $\sup_j\bigl|T-\mathcal{I}_{T}^{nj}\bigr|\xrightarrow[n\to\infty]{}0$ and the corresponding Riemann step functions $\sum_{j=1}^{m_n}(\beta\circ\Re)\bigl(T-t^{nj}\bigr)1_{T-\mathcal{I}_T^{nj}}$ and $\sum_{j=1}^{m_n}\bigl(\widetilde{\beta}\circ\Re\bigr)\bigl(T-t^{nj}\bigr)1_{T-\mathcal{I}_T^{nj}}$ approximate $\beta\circ\Re$ respectively $\widetilde{\beta}\circ\Re$, in $L^1(\mathcal{I}_T)$ as well as pointwise almost everywhere. Consequently, $\mathscr{U}_{\mathbb{A}'}(T-\tau,T-t)$ is also the strong limit in $\mathcal{L}(L^2)$ (uniformly in $(t,\tau)\in\mathcal{T}_T$) of a sequence $\bigl(\mathscr{U}_{\mathbb{A}'}^n(T-\tau,T-t)\bigr)_{n\in\mathbb{N}}$ satisfying
\begin{itemize}
\item $\mathscr{U}_{\mathbb{A}'}^n(T-\tau,T-t)=\mathrm{e}^{(t-\tau)\mathbb{A}(t^{nj})}$ for $t,\tau\in\overline{\mathcal{I}^{nj}_{T}}$ with $t\geq \tau$,
\item $\mathscr{U}_{\mathbb{A}'}^n(T-\tau,T-t)=\mathscr{U}_{\mathbb{A}'}^n(T-\tau,T-s)\mathscr{U}_{\mathbb{A}'}^n(T-s,T-t)$ for $t\geq s\geq \tau$.
\end{itemize}
Now, \eqref{bpoint} follows immediately from the four properties of $\mathscr{U}_{\mathbb{A}}^n$ and $\mathscr{U}_{\mathbb{A}'}^n$ listed above.
\end{bemaerkning}
We now apply Theorem \ref{Katoth} to the problem \eqref{LinSchr}--\eqref{LinSchrI}.
\begin{korollar}\label{existenceofevolution}
For all $T>0$ and all $\bm{A}\in W^{1,1}\bigl(\mathcal{I}_{T};L^4(\mathbb{R}^3;\mathbb{R}^3)\bigr)$ whose continuous representative is divergence free at all times there exists a unique evolution operator $\mathscr{U}_{\bm{A}}$ defined on $\mathcal{I}_{T}^2$ such that
\begin{enumerate}
\item[\emph{(A)}] $\mathscr{U}_{\bm{A}}$ is strongly continuous $\mathcal{I}_{T}^2\to\mathcal{L}(L^2)$ with $\mathscr{U}_{\bm{A}}(t,t)=1$ for $t\in\mathcal{I}_{T}$,
\item[\emph{(B)}] $\mathscr{U}_{\bm{A}}(t,\tau)\mathscr{U}_{\bm{A}}(\tau,s)=\mathscr{U}_{\bm{A}}(t,s)$ for all $(t,\tau,s)\in \mathcal{I}_{T}^3$,
\item[\emph{(C)}] $\mathscr{U}_{\bm{A}}(t,\tau)H^2\subset H^2$ for $(t,\tau)\in\mathcal{I}_{T}^2$ and $\mathscr{U}_{\bm{A}}:\mathcal{I}_{T}^2\to \mathcal{L}(H^2)$ is strongly continuous,
\item[\emph{(D)}] The strong partial derivatives $\partial_t\mathscr{U}_{\bm{A}}(t,\tau)\psi_0$ and $\partial_\tau \mathscr{U}_{\bm{A}}(t,\tau)\psi_0$ exist in $L^2$ for all $(t,\tau)\in \mathcal{I}_T^2$ and $\psi_0\in H^2$ and are given by
\begin{align*}
i\hbar\partial_t \mathscr{U}_{\bm{A}}(t,\tau)\psi_0=\Bigl(\sum_{j=1}^N \frac{1}{2m_j}\nabla_{j,\bm{A}(t)}^2+\sum_{1\leq j<k\leq N}\frac{Q_jQ_k}{|\bm{x}_j-\bm{x}_k|}\Bigr)\mathscr{U}_{\bm{A}}(t,\tau)\psi_0
\end{align*}
respectively
\begin{align*}
\hbar\partial_{\tau} \mathscr{U}_{\bm{A}}(t,\tau)\psi_0=i\mathscr{U}_{\bm{A}}(t,\tau)\Bigl(\sum_{j=1}^N \frac{1}{2m_j}\nabla_{j,\bm{A}(\tau)}^2+\sum_{1\leq j<k\leq N}\frac{Q_j Q_k}{|\bm{x}_j-\bm{x}_k|}\Bigr)\psi_0.
\end{align*}
Moreover, $\partial_t\mathscr{U}_{\bm{A}},\partial_{\tau}\mathscr{U}_{\bm{A}}:\mathcal{I}_{T}^2\to \mathcal{L}(H^2,L^2)$ are strongly continuous.
\end{enumerate}
\end{korollar}
\begin{proof}
Let $\bm{A}:\mathcal{I}_T\to L^4$ denote (the absolutely continuous representative of) a magnetic vector potential satisfying the hypotheses of the corollary and consider it's strong derivative $\partial_t\bm{A}$ that is defined almost everywhere on $\mathcal{I}_T$ and contained in $L^1(\mathcal{I}_T;L^4)$. Our goal will be to apply Theorem \ref{Katoth} and Remark \ref{Katore} to the family of operators
\begin{align*}
\mathbb{A}(t)=\frac{i}{\hbar}\Bigl(\sum_{j=1}^N \frac{1}{2m_j}\nabla_{j,\bm{A}(t)}^2+\sum_{1\leq j<k\leq N}\frac{Q_jQ_k}{|\bm{x}_j-\bm{x}_k|}\Bigr)
\end{align*}
in $\mathcal{X}=L^2\bigl(\mathbb{R}^{3N}\bigr)$ with domain $\mathcal{Y}=H^2\bigl(\mathbb{R}^{3N}\bigr)$. By Stone's theorem the selfadjointness of $i\mathbb{A}(t)$ implies that $-\mathbb{A}(t)$ generates a strongly continuous one-parameter group $\mathbb{R}\ni s\mapsto \exp\bigl(-s\mathbb{A}(t)\bigr)\in\mathcal{L}(L^2)$ of unitary operators for each $t\in\mathcal{I}_T$. Thereby $[0,\infty)\ni s\mapsto \exp\bigl(-s\mathbb{A}(t)\bigr)$ and $[0,\infty)\ni s\mapsto \exp\bigl(s\mathbb{A}(T-t)\bigr)$ are strongly continuous one-parameter semigroups generated by $-\mathbb{A}(t)$ respectively $\mathbb{A}(T-t)$. Moreover, the unitarity of the operators $\exp\bigl(-s\mathbb{A}(t)\bigr)$ for $t\in\mathcal{I}_T$ and $s\in\mathbb{R}$ ensures that both of the families $\{\mathbb{A}(t)\mid t\in\mathcal{I}_T\}$ and $\{-\mathbb{A}(T-t)\mid t\in\mathcal{I}_T\}$ are (quasi-)stable with the common stability index $(1,0)$. Thus, $\mathbb{A}$ and $-\mathbb{A}\circ\Re$ both satisfy the point (i') from Theorem \ref{Katoth}.

The operator $-i\mathbb{A}(t)$ in $L^2$ is selfadjoint and bounded from below, uniformly in $t$, by some constant $-M$ so by setting
\begin{align*}
\mathbb{S}(t)=M+1-i\mathbb{A}(t)\textrm{ for }t\in\mathcal{I}_T,
\end{align*}
we obtain a family of selfadjoint operators in $L^2$ that all have lower bounds $\geq 1$ and thereby map their common domain $H^2$ bijectively onto $L^2$. Lemma \ref{Estlemma} even gives that $\mathbb{S}(t)$ is bounded, when considered as an operator from the Hilbert space $H^2$ to the Hilbert space $L^2$, whereby it's inverse must also be bounded according to the bounded inverse theorem. Consequently, $\mathbb{S}(t)$ is an isomorphism $H^2\to L^2$ and the identity $\mathbb{S}(t)\mathbb{A}(t)\mathbb{S}(t)^{-1}=\mathbb{A}(t)$ holds by construction for all $t\in\mathcal{I}_T$. To show the final part of (ii'\,\!'\,\!') we define
\begin{align*}
\dot{\mathbb{S}}(t)=\sum_{j=1}^N \frac{Q_j}{\hbar m_j c}\partial_t\bm{A}(t)(\bm{x}_j)\cdot\nabla_{j,\bm{A}(t)}
\end{align*}
as an $\mathcal{L}(H^2,L^2)$-element for almost all points $t\in\mathcal{I}_T$ -- namely the points where $\partial_t\bm{A}$ is well-defined. Lemma \ref{Estlemma} and the strong measurability $\mathcal{I}_{T}\to L^4$ of $\bm{A}$ and $\partial_t\bm{A}$ allow us to conclude that $\mathbb{S}$ and $\dot{\mathbb{S}}$ are strongly measurable $\mathcal{I}_T\to \mathcal{L}(H^2,L^2)$ with the estimates 
\begin{align*}
\|\mathbb{S}(t)\|_{\mathcal{L}(H^2,L^2)}\lesssim 1+\|\bm{A}(t)\|_{L^4}^2, \bigl\|\dot{\mathbb{S}}(t)\bigr\|_{\mathcal{L}(H^2,L^2)}\lesssim \|\partial_t\bm{A}(t)\|_{L^4}(1+\|\bm{A}(t)\|_{L^4})
\end{align*}
holding true for almost all $t\in\mathcal{I}_T$. Consequently, $\mathbb{S}$ and $\dot{\mathbb{S}}$ are both Bochner integrable $\mathcal{I}_{T}\to\mathcal{L}(H^2,L^2)$ -- in fact, it follows from \eqref{AAcont} that $\mathbb{S}$ is continuous. Given an arbitrary $C_0^\infty(\mathcal{I}_T^\circ)$-function $g$ we now get
\begin{align}
&\int_{0}^T\dot{\mathbb{S}}(t)g(t)\,\mathrm{d}t\nonumber\\
&=\sum_{j=1}^N\frac{Q_j}{m_jc\hbar}\Bigl(i\hbar\int_{0}^T\partial_t\bm{A}(t)g(t)\,\mathrm{d}t(\bm{x}_j)\cdot\nabla_{\bm{x}_j}+\frac{Q_j}{c}\int_{0}^T\bigl(\partial_t\bm{A}\cdot\bm{A}\bigr)(t)g(t)\,\mathrm{d}t(\bm{x}_j)\Bigr)\nonumber\\
&=-\int_{0}^T\mathbb{S}(t)g'(t)\,\mathrm{d}t.\label{weakderiv}
\end{align}
where we use that $\bm{A}^2\in W^{1,1}\bigl(\mathcal{I}_T;L^2(\mathbb{R}^3)\bigr)$ with
\begin{align}
\partial_t\bm{A}^2(t)=2\partial_t\bm{A}(t)\cdot\bm{A}(t)\textrm{ for almost all }t\in\mathcal{I}_T,\label{derivofAsq}
\end{align}
which follows from approximating $\bm{A}$ in $W^{1,1}\bigl(\mathcal{I}_T;L^4(\mathbb{R}^3;\mathbb{R}^3)\bigr)$ by functions in the form $\bm{A}^n:t\mapsto\sum_{m=1}^{M^n}\bm{a}_m^nf_m^n(t)$ with $M^n\in\mathbb{N}$, $\bm{a}_1^n,\ldots,\bm{a}_{M^n}^n\in L^4(\mathbb{R}^3;\mathbb{R}^3)$ and $f_1^n,\ldots,f_{M^n}^n\in C^\infty(\mathcal{I}_T)$ for $n\in\mathbb{N}$. We conclude from \eqref{weakderiv} that the function $\mathbb{S}\in W^{1,1}\bigl(\mathcal{I}_T,\mathcal{L}(H^2,L^2)\bigr)$ has $\dot{\mathbb{S}}$ as it's derivative, whereby (ii'\,\!'\,\!') from Theorem \ref{Katoth} has been verified. 

Finally, we obtain from Lemma \ref{Estlemma} that for all $t,t'\in\mathcal{I}_{T}$
\begin{align}
\|\mathbb{A}(t)-\mathbb{A}(t')\|_{\mathcal{L}(H^2,L^2)}\lesssim \bigl(1+\|\bm{A}(t)+\bm{A}(t')\|_{L^4}\bigr)\|\bm{A}(t)-\bm{A}(t')\|_{L^4}\label{AAcont}
\end{align}
so the continuity of $\bm{A}:\mathcal{I}_{T}\to L^4$ implies that $\mathbb{A}:\mathcal{I}_{T}\to \mathcal{L}(H^2,L^2)$ is norm-continuous. Thus, also the point (iii) of Theorem \ref{Katoth} is satisfied.
\end{proof}
\begin{bemaerkning}\label{abscontU}
Let $\psi_0\in H^2$ and $\tau\in\mathcal{I}_{T}$ be given and set $\xi(t)=\mathscr{U}_{\bm{A}}(t,\tau)\psi_0$ for $t\in\mathcal{I}_{T}$. Being strongly differentiable $\mathcal{I}_{T}\to L^2$ with continuous derivative the function $\xi$ can be expressed as
\begin{align*}
\xi(t)= \xi(0) + \int_0^t \partial_t\xi(s)\,\mathrm{d}s\textrm{ for all }t\in\mathcal{I}_T,
\end{align*}
since the right hand side as a function of $t$ is strongly differentiable in $L^2$ with $\partial_t \xi$ as it's derivative by the mean value theorem. Thus, $\xi$ is absolutely continuous $\mathcal{I}_{T}\to L^2$, which in turn means that $\xi\in W^{1,1}(\mathcal{I}_{T};L^2)$ and that it's distributional derivative agrees with it's strong derivative.
\end{bemaerkning}
\begin{bemaerkning}\label{consL2norm}
By the same argument as in \eqref{consi} the mapping $\xi(t)=\mathscr{U}_{\bm{A}}(t,\tau)\psi_0$ has a conserved $L^2$-norm for any $\psi_0\in H^2$ and $\tau\in \mathcal{I}_T$. This together with the continuity of $\mathscr{U}_{\bm{A}}(t,\tau):L^2\to L^2$ implies that the $L^2$-norm of $\xi(t)$ is in fact a constant of the motion for all $\psi_0\in L^2$ and $\tau\in \mathcal{I}_T$.
\end{bemaerkning}
Given a potential $\bm{A}\in W^{1,1}\bigl(\mathcal{I}_T;L^4(\mathbb{R}^3;\mathbb{R}^3)\bigr)$ whose continuous representative is divergence free at all times we can according to Corollary \ref{existenceofevolution} apply $\mathscr{U}_{\bm{A}}(t,\tau)$ to any $L^2$-function $\psi_0$ and thereby obtain another $L^2$-function, even though we are only guaranteed that the result $\mathscr{U}_{\bm{A}}(t,\tau)\psi_0$ actually solves \eqref{LinSchr} if $\psi_0\in H^2$. However, by the estimates \eqref{Hminzwei} the right hand side of \eqref{LinSchr} is in fact meaningful (as an $H^{-2}$-element) when $\xi(t)$ is merely an $L^2$-function, provided that we interpret $\nabla_{j,\bm{A}(t)}^2\xi(t)$ as the sum
\begin{align}
-\hbar^2\Delta_{\bm{x}_j}\xi(t)+2i\frac{\hbar Q_j}{c}\mathrm{div}_{\bm{x}_j}\bigl(\bm{A}(t)(\bm{x}_j)\xi(t)\bigr)+\frac{Q_j^2}{c^2}\bigl[\bm{A}(t)(\bm{x}_j)\bigr]^2\xi(t).\label{Hm2inter}
\end{align}
A special case of the result below shows that for $\psi_0\in L^2$ there can not be any other $C(\mathcal{I}_{T};L^2)\cap W^{1,1}(\mathcal{I}_{T};H^{-2})$-solutions to the initial value problem \eqref{LinSchr}--\eqref{LinSchrI} than $\mathscr{U}_{\bm{A}}(t,\tau)\psi_0$. In order to formulate this result we introduce for $(t,\tau)\in\mathcal{I}_{T}^2$ the linear operator $H^{-2}\to H^{-2}$ (that we will again call $\mathscr{U}_{\bm{A}}(t,\tau)$) by setting
\begin{align*}
\langle \mathscr{U}_{\bm{A}}(t,\tau)\xi,\zeta\rangle_{H^{-2},H^2}=\langle \xi,\mathscr{U}_{\bm{A}}(\tau,t)\zeta\rangle_{H^{-2},H^2}
\end{align*}
for $\xi\in H^{-2}$ and $\zeta\in H^2$, where we remember that $H^{-s}$ is isometrically anti-isomorphic to the dual space $(H^s)^*$ of $H^s$ by the mapping
\begin{align*}
H^{-s}\ni \xi\mapsto \Bigl(\langle \xi,\cdot\rangle_{H^{-s},H^s}:\zeta\mapsto\frac{1}{(2\pi)^3} \bigl(\langle\bm{p}\rangle^{-s}\widehat{\xi},\langle\bm{p}\rangle^s\widehat{\zeta}\bigr)_{L^2}\Bigr)\in \bigl(H^s\bigr)^*.
\end{align*}
Then $\mathscr{U}_{\bm{A}}(t,\tau)$ is bounded with
\begin{align}
\|\mathscr{U}_{\bm{A}}(t,\tau)\|_{\mathcal{L}(H^{-2})}\leq \|\mathscr{U}_{\bm{A}}(\tau,t)\|_{\mathcal{L}(H^2)}\leq \sup_{(t',\tau')\in \mathcal{I}_{T}^2}\bigl\|\mathscr{U}_{\bm{A}}\bigl(t',\tau'\bigr)\bigr\|_{\mathcal{L}(H^2)},\label{normHminusto}
\end{align}
for $(t,\tau)\in\mathcal{I}_{T}^2$, where the right hand side is finite by the uniform boundedness principle. Moreover, $\mathscr{U}_{\bm{A}}(t,\tau):H^{-2}\to H^{-2}$ is an extension of the unitary operator $\mathscr{U}_{\bm{A}}(t,\tau): L^2\to L^2$ in the sense that they agree on $L^2$-functions. 
\begin{lemma}\label{integralformofequation}
Let the continuous representative of $\bm{A}\in W^{1,1}\bigl(\mathcal{I}_{T};L^4(\mathbb{R}^3;\mathbb{R}^3)\bigr)$ be divergence free at all times and consider some arbitrary $f\in L^1\bigl(\mathcal{I}_{T};H^{-2}\bigr)$. Then if $\xi\in C(\mathcal{I}_{T};L^2)\cap W^{1,1}\bigl(\mathcal{I}_{T};H^{-2}\bigr)$ satisfies the inhomogeneous many-body Schrödinger equation
\begin{align*}
i\hbar\partial_t \xi=\Bigl(\sum_{j=1}^N \frac{1}{2m_j}\nabla_{j,\bm{A}}^2+\sum_{1\leq j<k\leq N}\frac{Q_jQ_k}{|\bm{x}_j-\bm{x}_k|}\Bigr)\xi+f
\end{align*}
then
\begin{align}
\xi(t)=\mathscr{U}_{\bm{A}}(t,\tau)\xi(\tau)-\frac{i}{\hbar}\int_{\tau}^t \mathscr{U}_{\bm{A}}(t,s)f(s)\,\mathrm{d}s,\label{H-2ligning}
\end{align}
for all $(t,\tau)\in\mathcal{I}_{T}^2$.
\end{lemma}
\begin{proof}
Given some $t\in\mathcal{I}_T$ and $\zeta\in H^2$ the map $\langle \mathscr{U}_{\bm{A}}(t,\cdot)\xi(\cdot),\zeta\bigr\rangle_{H^{-2},H^2}$ is absolutely continuous since $\xi:\mathcal{I}_{T}\to H^{-2}$, $\mathscr{U}_{\bm{A}}(\cdot,t)\zeta:\mathcal{I}_{T}\to L^2$ are absolutely continuous (see Remark \ref{abscontU}) and $\xi:\mathcal{I}_{T}\to L^2$, $\mathscr{U}_{\bm{A}}(\cdot,t)\zeta:\mathcal{I}_{T}\to H^2$ are continuous. It's derivative is well defined almost everywhere in $\mathcal{I}_T$ and for almost all $s\in\mathcal{I}_T$
\begin{align}
\partial_{s}\bigl\langle \mathscr{U}_{\bm{A}}(t,s)\xi(s),\zeta\bigr\rangle_{H^{-2},H^2}&=\bigl\langle \partial_{s}\xi(s),\mathscr{U}_{\bm{A}}(s,t)\zeta\bigr\rangle_{H^{-2},H^2}+\bigl( \xi(s),\partial_{s}\mathscr{U}_{\bm{A}}(s,t)\zeta\bigr)_{L^2}\nonumber\\
&=\frac{i}{\hbar}\bigl\langle \mathscr{U}_{\bm{A}}(t,s)f(s),\zeta\bigr\rangle_{H^{-2},H^2},\label{Hminusto}
\end{align}
where $\bigl\langle \bigl(\sum_{j=1}^N \frac{1}{2m_j}\nabla_{j,\bm{A}(s)}^2+\sum_{j<k}\frac{Q_jQ_k}{|\bm{x}_j-\bm{x}_k|}\bigr)\xi(s),\mathscr{U}_{\bm{A}}(s,t)\zeta\bigr\rangle_{H^{-2},H^2}$ is seen to be equal to $\bigl(\xi(s),\bigl(\sum_{j=1}^N \frac{1}{2m_j}\nabla_{j,\bm{A}(s)}^2+\sum_{j < k}\frac{Q_jQ_k}{|\bm{x}_j-\bm{x}_k|}\bigr)\mathscr{U}_{\bm{A}}(s,t)\zeta\bigr)_{L^2}$ by approximating $\xi(s)$ and $\mathscr{U}_{\bm{A}}(s,t)\zeta$ in $L^2$ respectively $H^2$ by sequences of $C_0^\infty$-functions and using the estimates \eqref{KatoRellich} and \eqref{Hminzwei}. 
Thus,
\begin{align*}
\langle \xi(t),\zeta\rangle_{H^{-2},H^2}=\langle \mathscr{U}_{\bm{A}}(t,\tau)\xi(\tau),\zeta\rangle_{H^{-2},H^2}+\frac{i}{\hbar}\int_{\tau}^t\langle \mathscr{U}_{\bm{A}}(t,s)f(s),\zeta\rangle_{H^{-2},H^2}\,\mathrm{d}s
\end{align*}
for all $\tau\in\mathcal{I}_T$. Here, \eqref{normHminusto} and the assumption that $f\in L^1(\mathcal{I}_{T},H^{-2})$ give that $\mathscr{U}_{\bm{A}}(t,\cdot)f(\cdot)$ is Bochner integrable $\mathcal{I}_{T}\to H^{-2}$, whereby we can use \cite[Corollary V.5.2]{Yo} to commute the integral with the bounded anti-linear operator $\langle \cdot,\zeta\rangle_{H^{-2},H^2}:H^{-2}\to\mathbb{C}$ and obtain
\begin{align*}
\Bigl\langle \xi(t)-\mathscr{U}_{\bm{A}}(t,\tau)\xi(\tau)+\frac{i}{\hbar}\int_{\tau}^t\mathscr{U}_{\bm{A}}(t,s)f(s)\,\mathrm{d}s,\zeta\Bigr\rangle_{H^{-2},H^2}=0
\end{align*}
for all $\tau\in\mathcal{I}$, whereby the identity \eqref{H-2ligning} follows.
\end{proof}
As already mentioned in \eqref{normHminusto} the norms $\|\mathscr{U}_{\bm{A}}(t,\tau)\|_{\mathcal{L}(H^2)}$ are uniformly bounded in $(t,\tau)\in\mathcal{I}_{T}^2$. We will now find an explicit upper bound.
\begin{lemma}\label{H2wp}
Consider a vector potential $\bm{A}\in W^{1,1}\bigl(\mathcal{I}_{T};L^4(\mathbb{R}^3;\mathbb{R}^3)\bigr)$ whose continuous representative is divergence free at all times. Then for all $0<\delta<\frac{1}{2}$ there exists a constant $C>0$ \emph{(}depending on $c$, $\hbar$, $\delta$, $N$, $m_1,\ldots,m_N$ and $Q_1,\ldots,Q_N$\emph{)} such that
\begin{align}
\|\mathscr{U}_{\bm{A}}(t,\tau)\|_{\mathcal{L}(H^2)}\leq C\langle\|\bm{A}\|_{L_T^\infty L^4}\rangle^{\frac{8}{1-2\delta}} \exp\Bigl(C\int_{\tau}^t \langle \|\bm{A}(s)\|_{L^4}\rangle \|\partial_t\bm{A}(s)\|_{L^4}\,\mathrm{d}s\Bigr)\label{H2H2}
\end{align}
for all $(t,\tau)\in \mathcal{T}_T$.
\end{lemma}
\begin{proof}
Given $\psi_0\in H^2$ and $\tau\in \mathcal{I}_{T}$ we set $\xi(\cdot)=\mathscr{U}_{\bm{A}}(\cdot,\tau)\psi_0\in C(\mathcal{I}_{T};H^2)$ and note that the time derivative 
\begin{align*}
\partial_t\xi=-\frac{i}{\hbar}\Bigl(\sum_{j=1}^N \frac{1}{2m_j}\nabla_{j,\bm{A}}^2+\sum_{1\leq j<k\leq N}\frac{Q_jQ_k}{|\bm{x}_j-\bm{x}_k|}\Bigr)\xi
\end{align*}
has the distributional derivative given by
\begin{align}
\partial_t^2\xi&=-\frac{i}{\hbar}\Bigl(\sum_{j=1}^N\frac{1}{2m_j}\nabla_{j,\bm{A}}^2+\sum_{1\leq j<k\leq N}\frac{Q_jQ_k}{|\bm{x}_j-\bm{x}_k|}\Bigr)\partial_t\xi-\frac{i}{\hbar}f\label{timederiv}
\end{align}
where $\nabla_{j,\bm{A}}^2$ is interpreted as in \eqref{Hm2inter} and we introduce the $L^1(\mathcal{I}_T;L^2)$-map
\begin{align*}
f(t)=i\sum_{j=1}^N\frac{\hbar Q_j}{cm_j}\mathrm{div}_{\bm{x}_j}\bigl(\partial_t\bm{A}(t)(\bm{x}_j)\xi(t)\bigr)+\sum_{j=1}^N\frac{Q_j^2}{c^2 m_j}\bm{A}(t)(\bm{x}_j)\cdot\partial_t\bm{A}(t)(\bm{x}_j)\xi(t).
\end{align*}
This can be shown by approximating $\xi$ in $W^{1,1}\bigl(\mathcal{I}_{T};L^2\bigl(\mathbb{R}^{3N}\bigr)\bigr)$ by a sequence of maps $\xi^n:t\mapsto \sum_{m=1}^{M^n} \xi_m^n f_m^n(t)$ with $M^n\in\mathbb{N}$, $\xi_1^n,\ldots,\xi_{M^n}^n\in L^2\bigl(\mathbb{R}^{3N}\bigr)$ and $f_1^n,\ldots,f_{M^n}^n\in C^\infty(\mathcal{I}_T)$ for $n\in\mathbb{N}$. From \eqref{Hminzwei} and \eqref{derivofAsq} it follows for example that
\begin{align*}
&\int_0^T \mathrm{div}_{\bm{x}_j}\bigl(\bm{A}(t)(\bm{x}_j)\xi(t)\bigr)g'(t)\,\mathrm{d}t=\lim_{n\to\infty}\sum_{m=1}^{M^n}\mathrm{div}_{\bm{x}_j}\Bigl(\int_0^T \bm{A}(t)(f_m^n g')(t)\,\mathrm{d}t(\bm{x}_j)\xi_m^n\Bigr)\\
&=-\int_0^T \mathrm{div}_{\bm{x}_j}\bigl(\partial_t\bm{A}(t)(\bm{x}_j)\xi(t)+\bm{A}(t)(\bm{x}_j)\partial_t\xi(t)\bigr)g(t)\,\mathrm{d}t
\end{align*}
and
\begin{align*}
&\int_0^T \bigl[\bm{A}(t)(\bm{x}_j)\bigr]^2\xi(t) g'(t)\,\mathrm{d}t=\lim_{n\to\infty}\sum_{m=1}^{M^n}\int_0^T\bigl[\bm{A}(t)\bigr]^2(f_{m}^n g')(t)\,\mathrm{d}t(\bm{x}_j)\xi_{m}^n\\
&=-\int_0^T \bigl[\bm{A}(t)(\bm{x}_j)\bigr]^2\partial_t\xi(t)g(t)\,\mathrm{d}t-2\int_0^T \bm{A}(t)(\bm{x}_j)\cdot\partial_t\bm{A}(t)(\bm{x}_j)\xi(t)g(t)\,\mathrm{d}t,
\end{align*}
for all $j\in\{1,\ldots,N\}$ and $g\in C_0^\infty(\mathcal{I}_T^\circ)$, where the limits are taken in $H^{-2}$.
From \eqref{KatoRellich}, \eqref{Hminzwei}, \eqref{timederiv}, Corollary \ref{existenceofevolution} and Lemma \ref{integralformofequation} we get for all $t\in\mathcal{I}_{T}$ that
\begin{align*}
\partial_t\xi(t)=\mathscr{U}_{\bm{A}}(t,\tau)\partial_t\xi(\tau)-\frac{i}{\hbar}\int_{\tau}^t \mathscr{U}_{\bm{A}}(t,s)f(s)\,\mathrm{d}s.
\end{align*}
By using \eqref{KatoRellich} and Remark \ref{consL2norm} we therefore get the existence of a constant $K>0$ such that
\begin{align*}
&\|\xi(t)\|_{H^2}\\
&\leq K\!\left(\langle \|\bm{A}\|_{L_T^\infty L^4}\rangle^{\frac{8}{1-2\delta}}\|\xi(\tau)\|_{H^2}+\!\int_{\tau}^t \!\!\|\partial_t\bm{A}(s)\|_{L^4}\langle \|\bm{A}(s)\|_{L^4}\rangle\|\xi(s)\|_{H^2}\,\mathrm{d}s\right)
\end{align*}
for all $t\in[\tau,T]$ so \eqref{H2H2} holds by Gronwall's inequality.
\end{proof}

\chapter{The Klein-Gordon Equation}\label{KleinGo}
Given $\sigma\in\mathbb{R}$, $(\bm{A}_0,\bm{A}_1)\in H^{\sigma}\times H^{\sigma-1}$ and $\bm{F}\in L^1\bigl(\mathcal{I}_T;H^{\sigma-1}\bigr)$ define the continuous function $\mathscr{V}_{\bm{F}}(\cdot,0)[\bm{A}_0,\bm{A}_1]:\mathcal{I}_T\to H^{\sigma}$ by
\begin{align}
\mathscr{V}_{\bm{F}}(t,0)[\bm{A}_0,\bm{A}_1]=\dot{\mathfrak{s}}(t)\bm{A}_0+\mathfrak{s}(t)\bm{A}_1+c^2\int_0^t \mathfrak{s}(t-\tau)\bm{F}(\tau)\,\mathrm{d}\tau,\label{KGsol}
\end{align}
where the two linear operators $\dot{\mathfrak{s}}(t)=\cos\bigl(c(1-\Delta)^{1/2} t\bigr):H^{\sigma}\to H^{\sigma}$ and $\mathfrak{s}(t)=\frac{\sin(c(1-\Delta)^{1/2} t)}{c(1-\Delta)^{1/2}}:H^{\sigma-1}\to H^{\sigma}$ are defined as Fourier multipliers for $t\in\mathcal{I}_T$. Then $\mathscr{V}_{\bm{F}}(\cdot,0)[\bm{A}_0,\bm{A}_1]$ has the $C(\mathcal{I}_T;H^{\sigma-1})$-mapping
\begin{align*}
\partial_t\mathscr{V}_{\bm{F}}(t,0)[\bm{A}_0,\bm{A}_1]=c^2(\Delta-1)\mathfrak{s}(t)\bm{A}_0+\dot{\mathfrak{s}}(t)\bm{A}_1+c^2\int_0^t \dot{\mathfrak{s}}(t-\tau)\bm{F}(\tau)\,\mathrm{d}\tau
\end{align*}
as distributional first derivative and the $L^1(\mathcal{I}_T;H^{\sigma-2})$-function
\begin{align*}
\partial_t^2\mathscr{V}_{\bm{F}}(t,0)[\bm{A}_0,\bm{A}_1]=c^2(\Delta-1)\mathscr{V}_{\bm{F}}(t,0)[\bm{A}_0,\bm{A}_1]+c^2\bm{F}(t).
\end{align*}
as distributional second derivative. In other words, $\mathscr{V}_{\bm{F}}(\cdot,0)[\bm{A}_0,\bm{A}_1]$ solves the Klein-Gordon equation
\begin{align}
(\Box+1)\bm{B}=\bm{F}\label{KG}
\end{align}
with initial conditions
\begin{align}
\bm{B}(0)=\bm{A}_0\textrm{ and }\partial_t\bm{B}(0)=\bm{A}_1.\label{KGinit}
\end{align}
As expressed below in Lemma \ref{KleinGordon} the function \eqref{KGsol} can be shown to be a $C(\mathcal{I}_T;H^{\sigma})\cap C^1(\mathcal{I}_T;H^{\sigma-1})$-solution to the initial value problem \eqref{KG}--\eqref{KGinit} for even more general choices of inhomogeneity $\bm{F}$. We will need the accompanying Strichartz estimate. The result is due to Brenner \cite{Brenner}, Strichartz \cite{Strichartz}, Ginibre and Velo \cite{Gin-1,Gin0}, but is formulated on the basis of \cite[Lemma 4.1]{NW}.
\begin{lemma}{\emph{\cite[Lemma 4.1]{NW}}}\label{KleinGordon}
Let $0\leq \frac{2}{q_k}= 1-\frac{2}{r_k}<1$ for $k\in\{0,1\}$. Then for $\sigma\in\mathbb{R}$, $(\bm{A}_0,\bm{A}_1)\in H^{\sigma}\times H^{\sigma-1}$ and $\bm{F}\in L^{q_1'}\bigl(\mathcal{I}_T;W^{\sigma-1+\frac{2}{q_1},r_1'}\bigr)$ the function $\bm{B}(t)=\mathscr{V}_{\bm{F}}(\cdot,0)[\bm{A}_0,\bm{A}_1]$ in \eqref{KGsol} is contained in $C(\mathcal{I}_T;H^{\sigma})\cap C^1(\mathcal{I}_T;H^{\sigma-1})$ and the Strichartz estimate
\begin{align*}
\max_{k\in\{0,1\}}\|\partial_t^k\bm{B}\|_{L_T^{q_0}W^{\sigma-k-\frac{2}{q_0},r_0}}\lesssim \|(\bm{A}_0,\bm{A}_1)\|_{H^{\sigma}\times H^{\sigma-1}}+\|\bm{F}\|_{L_T^{q_1'}W^{\sigma-1+\frac{2}{q_1},r_1'}}
\end{align*}
holds true.
\end{lemma}

\chapter{The Contraction Argument}\label{Contraction}
Let $(\psi_0,\bm{A}_0,\bm{A}_1)\in H^2(\mathbb{R}^{3N})\times H^{\frac{3}{2}}(\mathbb{R}^3;\mathbb{R}^3)\times H^{\frac{1}{2}}(\mathbb{R}^3;\mathbb{R}^3)$ satisfy the identities $\mathrm{div}\bm{A}_0=\mathrm{div}\bm{A}_1=0$ and consider for $T,R_1,R_2\in(0,\infty)$ the mapping $\Phi$ sending a pair $(\psi,\bm{A})$ from the $(T,R_1,R_2)$-dependent space
\begin{align*}
\mathcal{Z}_T&\!=\!\bigl\{(\psi,\bm{A})\in L^\infty(\mathcal{I}_{T};H^2)\!\times\! \bigl(L^\infty \bigl(\mathcal{I}_{T};H^1(\mathbb{R}^3;\mathbb{R}^3)\bigr)\cap W^{1,4}\bigl(\mathcal{I}_{T};L^4(\mathbb{R}^3;\mathbb{R}^3)\bigr)\bigr)\big| \nonumber\\
&\hspace{0.5cm}\textrm{the continuous representative }\mathcal{I}_T\to L^4(\mathbb{R}^3;\mathbb{R}^3)\textrm{ of }\bm{A}\textrm{ is divergence free}\\
&\hspace{0.5cm}\textrm{at all times},\|\psi\|_{L_T^\infty H^2}\leq R_1,\max\{\|\bm{A}\|_{L_T^\infty H^1},\|\bm{A}\|_{W_T^{1,4}L^4}\}\leq R_2\bigr\}
\end{align*}
into the solution $\Phi(\psi,\bm{A})=\bigl(\mathscr{U}_{\bm{A}}(\cdot,0)\psi_0,\mathscr{V}_{\frac{4\pi}{c}\sum_{j=1}^N P\bm{J}_j[\psi,\bm{A}]+\bm{A}}(\cdot,0)[\bm{A}_0,\bm{A}_1]\bigr)$ to the linearized system
\begin{align}
i\hbar\partial_t\xi&= \Bigl(\sum_{j=1}^N \frac{1}{2m_j}\nabla_{j,\bm{A}}^2+\sum_{1 \leq j < k \leq N}\frac{Q_jQ_k}{|\bm{x}_j-\bm{x}_k|}\Bigr)\xi,\label{SE}\\
(\Box+1)\bm{B}&=\frac{4\pi}{c}\sum_{j=1}^N P\bm{J}_j[\psi,\bm{A}]+\bm{A}\label{KGE}
\end{align}
with initial data
\begin{align*}
\xi(0)=\psi_0,\quad \bm{B}(0)=\bm{A}_0\quad\textrm{and}\quad\partial_t\bm{B}(0)=\bm{A}_1,
\end{align*}
where we observe that $W^{1,4}\bigl(\mathcal{I}_T;L^4(\mathbb{R}^3;\mathbb{R}^3)\bigr)\hookrightarrow W^{1,1}\bigl(\mathcal{I}_T;L^4(\mathbb{R}^3;\mathbb{R}^3)\bigr)$ and $P\bm{J}_j[\psi,\bm{A}]\in L^\infty(\mathcal{I}_T;H^1(\mathbb{R}^3;\mathbb{R}^3))$ for $j\in\{1,\ldots,N\}$ by \eqref{current} and the boundedness of the Helmholtz projection $H^1\to H^1$. 
Combining Corollary \ref{existenceofevolution} with Lemma \ref{KleinGordon} gives that $\Phi(\psi,\bm{A})\in C(\mathcal{I}_T;H^2)\times \bigl(C(\mathcal{I}_T;H^{\frac{3}{2}})\cap C^1(\mathcal{I}_T;H^{\frac{1}{2}})\bigr)$ and we observe directly from \eqref{KGsol} that the second coordinate of $\Phi(\psi,\bm{A})$ must be divergence free at all times, whereby a fixed point of $\Phi$ will have the desired properties. Our strategy will therefore be to invoke the Banach fixed-point theorem and for this we equip $\mathcal{Z}_T$ with the metric $d$ given by
\begin{align*}
d\bigl((\psi,\bm{A}),(\psi',\bm{A}')\bigr)=\max\bigl\{\|\psi-\psi'\|_{L_T^\infty L^2},\|\bm{A}-\bm{A}'\|_{L_T^\infty H^{\frac{1}{2}}},\|\bm{A}-\bm{A}'\|_{L_T^4L^4}\bigr\}
\end{align*}
for $(\psi,\bm{A}),(\psi',\bm{A}')\in \mathcal{Z}_T$.
\begin{lemma}
For all choices of positive numbers $T$, $R_1$ and $R_2$ the metric space $(\mathcal{Z}_T,d)$ is complete.
\end{lemma}
\begin{proof}
Let $\bigl((\psi_n,\bm{A}_n)\bigr)_{n\in\mathbb{N}}$ be a Cauchy sequence in $(\mathcal{Z}_T,d)$. Then $(\psi_n)_{n\in\mathbb{N}}$ is a Cauchy sequence in the Banach space $L^\infty(\mathcal{I}_{T};L^2)$ and $(\psi_n)_{n\in\mathbb{N}}$ is furthermore known to be bounded by the constant $R_1$ in the space $L^\infty(\mathcal{I}_{T};H^2)$ -- a space that can be identified with the dual of the separable space $L^1(\mathcal{I}_{T};H^{-2})$ by the isometric anti-isomorphism
\begin{align*}
L^\infty(\mathcal{I}_T;H^2)\ni F\mapsto \Bigl(G\mapsto\int_0^T \langle F(t),G(t)\rangle_{H^2,H^{-2}}\,\mathrm{d}t\Bigr)\in \bigl(L^1(\mathcal{I}_T;H^{-2})\bigr)^*
\end{align*}
as expressed in \cite[Theorem 8.18.3]{Edwards}. Therefore we can use the Banach-Alaoglu theorem to conclude that there exist $\psi\in L^\infty(\mathcal{I}_{T};L^2)$ and $\psi^*\in L^\infty(\mathcal{I}_{T};H^2)$ such that
\begin{align}
\psi_n\xrightarrow[n\to\infty]{}\psi\textrm{ in }L^\infty(\mathcal{I}_{T};L^2)\textrm{ and }\psi_{n_k}\xrightharpoonup[k\to\infty]{w*}\psi^*\textrm{ in }L^\infty(\mathcal{I}_{T};H^2).\label{psicon}
\end{align}
For $\varphi\in L^2(\mathcal{I}_{T};L^2)$ the sequence $\bigl((\psi_{n_k},\varphi)_{L^2L^2}\bigr)_{k\in\mathbb{N}}$ then converges to both of the numbers $(\psi,\varphi)_{L^2L^2}$ and $(\psi^*,\varphi)_{L^2L^2}$ so the functions $\psi$ and $\psi^*$ must be identical. Likewise, $(\bm{A}_n)_{n\in\mathbb{N}}$ is bounded by the constant $R_2$ in the dual $L^\infty(\mathcal{I}_{T};H^1)$ of the separable space $L^1(\mathcal{I}_{T};H^{-1})$ and in addition $(\bm{A}_n)_{n\in\mathbb{N}}$ is a Cauchy sequence in each of the two Banach spaces $L^\infty\bigl(\mathcal{I}_{T};H^{\frac{1}{2}}\bigr)$ and $L^4(\mathcal{I}_{T};L^4)$. Consequently, there exists an $\bm{A}\in L^\infty(\mathcal{I}_T,H^1)\cap L^4(\mathcal{I}_T,L^4)$ such that
\begin{align}
\bm{A}_n\xrightarrow[n\to\infty]{}\bm{A}\textrm{ in }L^\infty\bigl(\mathcal{I}_{T};H^{\frac{1}{2}}\bigr)\textrm{ and }L^4(\mathcal{I}_{T};L^4),\bm{A}_{n_k'}\xrightharpoonup[k\to\infty]{w*}\bm{A}\textrm{ in }L^\infty(\mathcal{I}_{T};H^1).\label{Akonverse}
\end{align}
Moreover, the boundedness of the sequence $(\partial_t\bm{A}_n)_{n\in\mathbb{N}}$ in the reflexive space $L^4(\mathcal{I}_{T};L^4)$ gives the existence of an $\dot{\bm{A}}\in L^4(\mathcal{I}_T;L^4)$ such that the weak convergence
\begin{align}
\partial_t\bm{A}_{n_k''}\xrightharpoonup[k\to\infty]{}\dot{\bm{A}}\textrm{ in }L^4(\mathcal{I}_{T};L^4)\label{Acon}
\end{align}
holds. But for any $k\in\mathbb{N}$, $\eta\in L^{\frac{4}{3}}(\mathbb{R}^3)$ and $\varphi\in C_0^\infty(\mathcal{I}_T^\circ)$ we then have
\begin{align*}
\int_0^T \int_{\mathbb{R}^3}\bm{A}_{n_k''}(t)(\bm{x})\eta(\bm{x})\,\mathrm{d}\bm{x}\varphi'(t)\,\mathrm{d}t=-\int_0^T\int_{\mathbb{R}^3}\partial_t\bm{A}_{n_k''}(t)(\bm{x})	\eta(\bm{x})\,\mathrm{d}\bm{x}\varphi(t)\,\mathrm{d}t
\end{align*}
whereby letting $k\to\infty$ and using \eqref{Akonverse}--\eqref{Acon} gives that $\dot{\bm{A}}$ is the distributional time derivative of $\bm{A}$. Concerning the divergence of $\bm{A}$ we observe that
\begin{align*}
\int_0^T \|\mathrm{div}\bm{A}(t)\|_{W^{-1,4}}^4\,\mathrm{d}t
\leq \|\bm{A}-\bm{A}_n\|_{L_T^4L^4}^4\xrightarrow[n\to\infty]{}0
\end{align*}
so $\mathrm{div}\bm{A}(t)=\mathrm{div}\partial_t\bm{A}(t)=0$ for almost all $t\in\mathcal{I}_T$. For any $t\in(0,T]$ the continuous representative $\mathcal{I}_T\to L^4(\mathbb{R}^3;\mathbb{R}^3)$ of $\bm{A}$ therefore satisfies
\begin{align*}
\mathrm{div}\bm{A}(t)=\mathrm{div}\bm{A}(t')+\int_{t'}^t\mathrm{div}\partial_t\bm{A}(s)\,\mathrm{d}s=0,
\end{align*}
where we have chosen some time $t'\in [0,t]$ in which $\bm{A}$ takes a divergence free value -- the identity $\mathrm{div}\bm{A}(0)=0$ then follows by using the continuity of $\mathcal{I}_T\ni t\mapsto \mathrm{div}\bm{A}(t)\in W^{-1,4}(\mathbb{R}^3)$. Finally, \cite[Propositions 3.5 and 3.13]{HB} concerning boundedness of weakly (respectively weak-$*$) convergent sequences combined with \eqref{psicon}--\eqref{Acon} give
\begin{align*}
\|\psi\|_{L_T^\infty H^2}\leq R_1\textrm{ and }\max\bigl\{\|\bm{A}\|_{L_T^\infty H^1},\|\bm{A}\|_{W_T^{1,4}L^4}\bigr\}\leq R_2,
\end{align*}
whereby we are in position to conclude that $(\psi,\bm{A})$ is contained in $\mathcal{Z}_T$ and that $d\bigl((\psi,\bm{A}),(\psi_n,\bm{A}_n)\bigr)\xrightarrow[n\to\infty]{}0$.
\end{proof}
Next, we investigate the properties of the mapping $\Phi$.
\begin{lemma}\label{into}
Given any $(\psi_0,\bm{A}_0,\bm{A}_1)\in H^{2}(\mathbb{R}^{3N})\times H^{\frac{3}{2}}(\mathbb{R}^3;\mathbb{R}^3)\times H^{\frac{1}{2}}(\mathbb{R}^3;\mathbb{R}^3)$ satisfying $\mathrm{div}\bm{A}_0=\mathrm{div}\bm{A}_1=0$ and any $R>0$ there exist $R_1,R_2\in (R,\infty)$ and $T_\dagger>0$ such that for all $T\in(0,T_\dagger]$ the function $\Phi$ maps $\mathcal{Z}_T$ into itself.
\end{lemma}
\begin{proof}
Let $(\psi_0,\bm{A}_0,\bm{A}_1)\in H^{2}(\mathbb{R}^{3N})\times H^{\frac{3}{2}}(\mathbb{R}^3;\mathbb{R}^3)\times H^{\frac{1}{2}}(\mathbb{R}^3;\mathbb{R}^3)$ satisfy the identities $\mathrm{div}\bm{A}_0=\mathrm{div}\bm{A}_1=0$ and consider arbitrary positive constants $T,R_1$ and $R_2$. For any fixed pair $(\psi,\bm{A})\in \mathcal{Z}_T$ we get from Lemma \ref{H2wp}, the Sobolev embedding $H^{\frac{3}{4}}\hookrightarrow L^4$, Lemma \ref{KleinGordon} and \eqref{current} that not only is $\Phi(\psi,\bm{A})=(\xi,\bm{B})$ contained in $C(\mathcal{I}_T;H^2)\times \bigl(C(\mathcal{I}_T;H^{\frac{3}{2}})\cap C^1(\mathcal{I}_T;H^{\frac{1}{2}})\bigr)$ as noted above, but we also have $\bm{B}\in W^{1,4}(\mathcal{I}_{T},L^4)$ with the two estimates
\begin{align*}
\|\xi\|_{L_T^\infty H^2}\leq C\langle R_2 \rangle^{\frac{8}{1-2\delta}} \exp\bigl(C T^{\frac{3}{4}} \langle R_2\rangle R_2\bigr)\|\psi_0\|_{H^2}
\end{align*}
and
\begin{align*}
&\max\bigl\{\|\bm{B}\|_{L_T^\infty H^{\frac{3}{2}}},\|\bm{B}\|_{L_T^4 L^4},\|\partial_t \bm{B}\|_{L_T^4 L^4}\bigr\}\\
&\leq C\bigl(\|(\bm{A}_0,\bm{A}_1)\|_{H^{\frac{3}{2}}\times H^{\frac{1}{2}}}+T(1+R_2)R_1^2+TR_2\bigr)
\end{align*}
holding true for some constant $C>0$ (depending on $c$, $\hbar$, $N$, $m_1,\ldots,m_N$ and $Q_1,\ldots,Q_N$). Given some $(\psi_0,\bm{A}_0,\bm{A}_1)\in H^{2}(\mathbb{R}^{3N})\times H^{\frac{3}{2}}(\mathbb{R}^3;\mathbb{R}^3)\times H^{\frac{1}{2}}(\mathbb{R}^3;\mathbb{R}^3)$ with $\mathrm{div}\bm{A}_0=\mathrm{div}\bm{A}_1=0$ and some positive number $R$ we can therefore choose $R_2>\max\bigl\{2\sqrt{2}C\|(\bm{A}_0,\bm{A}_1)\|_{H^{\frac{3}{2}}\times H^{\frac{1}{2}}},R\bigr\}$, $R_1>\max\bigl\{2C\langle R_2\rangle^{\frac{8}{1-2\delta}}\|\psi_0\|_{H^2},R\bigr\}$ and $T_\dagger=\min\bigl\{\frac{R_2}{2\sqrt{2}C((1+R_2)R_1^2+R_2)},\bigl(\frac{\log 2}{CR_2\langle R_2\rangle}\bigr)^{\frac{4}{3}}\bigr\}$ to make sure that $\Phi$ maps $\mathcal{Z}_T$ into itself for any $T\in (0,T_\dagger]$.
\end{proof}
Finally, we show that by choosing $T$ sufficiently small we can make $\Phi$ a contraction on $(\mathcal{Z}_T,d)$, which by the Banach fixed-point theorem guarantees the existence of a unique fixed point for $\Phi$.
\begin{lemma}\label{fixedpoint}
For any $(\psi_0,\bm{A}_0,\bm{A}_1)\in H^{2}(\mathbb{R}^{3N})\times H^{\frac{3}{2}}(\mathbb{R}^3;\mathbb{R}^3)\times H^{\frac{1}{2}}(\mathbb{R}^3;\mathbb{R}^3)$ with $\mathrm{div}\bm{A}_0=\mathrm{div}\bm{A}_1=0$ and any $R\geq 0$ there exist $R_1,R_2\in (R,\infty)$ and $T_*>0$ such that $\Phi$ is a contraction on $\bigl(\mathcal{Z}_T,d\bigr)$ for all $T\in(0,T_*]$.
\end{lemma}
\begin{proof}
Given $R\geq 0$ and $(\psi_0,\bm{A}_0,\bm{A}_1)\in H^{2}(\mathbb{R}^{3N})\times H^{\frac{3}{2}}(\mathbb{R}^3;\mathbb{R}^3)\times H^{\frac{1}{2}}(\mathbb{R}^3;\mathbb{R}^3)$ satisfying $\mathrm{div}\bm{A}_0=\mathrm{div}\bm{A}_1=0$ we use Lemma \ref{into} to choose $R_1,R_2\in(R,\infty)$ and $T_\dagger>0$ such that $\Phi$ maps $\mathcal{Z}_T$ into itself for any time span $T\in(0,T_\dagger]$. Given an arbitrary such $T\in(0,T_\dagger]$ we consider $(\psi,\bm{A}),(\psi',\bm{A}')\in\mathcal{Z}_T$ and write $\Phi(\psi,\bm{A})=(\xi,\bm{B})$ as well as $\Phi(\psi',\bm{A}')=(\xi',\bm{B}')$. After introducing $f\in C(\mathcal{I}_T;L^2)$ by setting
\begin{align*}
f(t)=\sum_{j=1}^N\frac{1}{2m_j}\bigl(\nabla_{j,\bm{A}(t)}^2-\nabla_{j,\bm{A}'(t)}^2\bigr)\xi'(t)\textrm{ for }t\in\mathcal{I}_T
\end{align*}
we observe that $\xi-\xi'$ solves the initial value problem
\begin{align*}
i\hbar \partial_t(\xi-\xi')&=\Bigl(\sum_{j=1}^N\frac{1}{2m_j}\nabla_{j,\bm{A}}^2+\sum_{1\leq j<k\leq N}\frac{Q_jQ_k}{|\bm{x}_j-\bm{x}_k|}\Bigr)(\xi-\xi')+f\\
(\xi-\xi')(0)&=0.
\end{align*}
Combining this with Lemma \ref{integralformofequation} gives that $(\xi-\xi')(t)=-\frac{i}{\hbar}\int_0^t \mathscr{U}_{\bm{A}}(t,s)f(s)\,\mathrm{d}s$ for all $t\in\mathcal{I}_{T}$, whereby Remark \ref{consL2norm}, Lemma \ref{Estlemma} and Hölder's inequality help us obtain the estimate
\begin{align}
\|\xi-\xi'\|_{L_T^\infty L^2}
&\lesssim \int_0^T (1+\|\bm{A}(s)+\bm{A}'(s)\|_{L^4})\|\bm{A}(s)-\bm{A}'(s)\|_{L^4}\|\xi'(s)\|_{H^2}\,\mathrm{d}s\nonumber\\
&\leq R_1\bigl(T^{\frac{3}{4}}+2R_2 T^{\frac{1}{2}}\bigr)\|\bm{A}-\bm{A}'\|_{L_T^4 L^4}.\label{xidist}
\end{align}
The map $\bm{B}-\bm{B}'=\mathscr{V}_{\frac{4\pi}{c}\sum_{j=1}^N P(\bm{J}_j[\psi,\bm{A}]-\bm{J}_j[\psi',\bm{A}'])}(\cdot,0)[\bm{0},\bm{0}]+\mathscr{V}_{\bm{A}-\bm{A}'}(\cdot,0)[\bm{0},\bm{0}]$ satisfies
\begin{align}
&\max\bigl\{\|\bm{B}-\bm{B}'\|_{L_T^\infty H^{\frac{1}{2}}},\|\bm{B}-\bm{B}'\|_{L_T^4L^4}\bigr\}\nonumber\\
&\lesssim \sum_{j=1}^N\bigl\| P\bigl(\bm{J}_j[\psi,\bm{A}]-\bm{J}_j[\psi',\bm{A}']\bigr)\bigr\|_{L_T^{\frac{4}{3}}L^{\frac{4}{3}}}+\|\bm{A}-\bm{A}'\|_{L_T^1 H^{-\frac{1}{2}}}\label{Bdist}
\end{align}
by Lemma \ref{KleinGordon}. To estimate the first term on the right hand side of \eqref{Bdist} we write $\bigl(\bm{J}_j[\psi,\bm{A}]-\bm{J}_j[\psi',\bm{A}']\bigr)(t)$ for almost all $t\in\mathcal{I}_T$ as a sum of the three $L^{\frac{4}{3}}$-functions
\begin{align*}
&g_j^1(t): \bm{x}_j\mapsto \frac{Q_j}{m_j}\mathrm{Re}\int_{\mathbb{R}^{3(N-1)}}\overline{(\psi'-\psi)}(t)(\bm{x})\nabla_{j,\bm{A}(t)}\psi(t)(\bm{x})\,\mathrm{d}\bm{x}_j',\\
&g_j^2(t): \bm{x}_j\mapsto \frac{Q_j^2}{m_j c}(\bm{A}'-\bm{A})(t)(\bm{x}_j)\mathrm{Re}\int_{\mathbb{R}^{3(N-1)}}\bigl(\overline{\psi}'\psi\bigr)(t)(\bm{x})\,\mathrm{d}\bm{x}_j'
\end{align*}
and
\begin{align}
&g_j^3(t):\bm{x}_j\mapsto \Bigl\{-\frac{Q_j}{m_j}\mathrm{Re}\int_{\mathbb{R}^{3(N-1)}}\nabla_{j,-\bm{A}'(t)}\overline{\psi}'(t)(\bm{x})(\psi'-\psi)(t)(\bm{x})\,\mathrm{d}\bm{x}_j'\nonumber\\
&\hspace{3.4cm}-\frac{Q_j\hbar}{m_j}\nabla_{\bm{x}_j}\mathrm{Im}\int_{\mathbb{R}^{3(N-1)}}\bigl[\overline{\psi}'(\psi'-\psi)\bigr](t)(\bm{x})\,\mathrm{d}\bm{x}_j'\Bigr\}.\label{gj3}
\end{align}
where the expression for the third function can also be written more compactly as $\bm{x}_j\mapsto \frac{Q_j}{m_j}\mathrm{Re}\int_{\mathbb{R}^{3(N-1)}}\overline{\psi}'(t)(\bm{x})\nabla_{j,\bm{A}'(t)}(\psi'-\psi)(t)(\bm{x})\,\mathrm{d}\bm{x}_j'$. However, in the present context we prefer to express $g_j^3$ in the form \eqref{gj3} since applying the Helmholtz projection kills the last term in \eqref{gj3} and leaves us with a term with no derivatives applied to the difference $(\psi'-\psi)(t)$. As in the proof of Lemma \ref{Estlemma} we can therefore use Minkowski's integral inequality, the Sobolev embeddings $H^{\frac{3}{4}}\hookrightarrow L^4$, $H^{\frac{3}{2}+\delta}\hookrightarrow L^\infty$, boundedness of the Helmholtz projection $L^{\frac{4}{3}}\to L^{\frac{4}{3}}$ and Hölder's inequality to obtain that for almost all $t\in\mathcal{I}_T$,
\begin{align*}
&\bigl\| P\bigl(\bm{J}_j[\psi,\bm{A}]-\bm{J}_j[\psi',\bm{A}']\bigr)(t)\bigr\|_{L^{\frac{4}{3}}}\\
&\lesssim \|g_j^1(t)\|_{L^{\frac{4}{3}}}+\|g_j^2(t)\|_{L^{\frac{4}{3}}}+\Bigl\|\bm{x}_j\mapsto\int\nabla_{j,-\bm{A}'(t)}\overline{\psi}'(t)(\bm{x})(\psi'-\psi)(t)(\bm{x})\,\mathrm{d}\bm{x}_j'\Bigr\|_{L^{\frac{4}{3}}}\nonumber\\
&\lesssim \bigl\{(1+\|\bm{A}(t)\|_{L^4})\|\psi(t)\|_{H^2}+(1+\|\bm{A}'(t)\|_{L^4})\|\psi'(t)\|_{H^2}\bigr\}\|(\psi'-\psi)(t)\|_{L^2}\\
&\hspace{5.5cm}+\|\psi'(t)\|_{L^2}\|\psi(t)\|_{H^2}\|(\bm{A}'-\bm{A})(t)\|_{L^4}.
\end{align*}
and so
\begin{align}
&\bigl\| P\bigl(\bm{J}_j[\psi,\bm{A}]-\bm{J}_j[\psi',\bm{A}']\bigr)\bigr\|_{L^{\frac{4}{3}}L^{\frac{4}{3}}}\nonumber\\
&\lesssim R_1\bigl(T^{\frac{3}{4}}+R_2 T^{\frac{1}{2}}\bigr)\|\psi'-\psi\|_{L_T^\infty L^2}+R_1^2T^{\frac{1}{2}}\|\bm{A}'-\bm{A}\|_{L_T^4L^4}.\label{Bvidere}
\end{align}
From \eqref{xidist}, \eqref{Bdist} and \eqref{Bvidere} we realize that there exists a constant $C>0$ such that
\begin{align*}
d\bigl(\Phi(\psi,\bm{A}),\Phi(\psi',\bm{A}')\bigr)\leq C\bigl(R_1\bigl(T^{\frac{3}{4}}+R_2 T^{\frac{1}{2}}\bigr)+R_1^2T^{\frac{1}{2}}+T\bigr) d\bigl((\psi,\bm{A}),(\psi',\bm{A}')\bigr)
\end{align*}
so for small enough $T$ the mapping $\Phi$ will be a contraction on $(\mathcal{Z}_T,d)$.
\end{proof}
The existence part of Theorem \ref{main} has now been proven. 

\chapter{Uniqueness}\label{Uniqueness}
We now turn our attention to the uniqueness question.
\begin{lemma}\label{uniq1}
Let $(\psi_0,\bm{A}_0,\bm{A}_1)\in H^2(\mathbb{R}^{3N})\times H^{\frac{3}{2}}(\mathbb{R}^3;\mathbb{R}^3)\times H^{\frac{1}{2}}(\mathbb{R}^3;\mathbb{R}^3)$ with $\mathrm{div}\bm{A}_0=\mathrm{div}\bm{A}_1=0$ and $T>0$ be given. Then if the pairs $(\psi^1,\bm{A}^1)$ and $(\psi^2,\bm{A}^2)$ belong to $C(\mathcal{I}_{T};H^2)\times \bigl(C\bigl(\mathcal{I}_T;H^{\frac{3}{2}}\bigr)\cap C^1\bigl(\mathcal{I}_T;H^{\frac{1}{2}}\bigr)\bigr)$, solve \eqref{MSPmaerke}$+$\eqref{initialcondi} and both of the vector fields $\bm{A}^1$, $\bm{A}^2$ are divergence free at all times in $[0,T]$ then there exists a $T_*\in(0,T]$ such that $(\psi^1,\bm{A}^1)$ and $(\psi^2,\bm{A}^2)$ agree on the time interval $[0,T_*]$.
\end{lemma}
\begin{proof}
For $\ell\in\{1,2\}$ let $(\psi^\ell,\bm{A}^\ell)$ satisfy the hypotheses of the lemma and choose with the help of Lemma \ref{fixedpoint} some radii
\begin{align*}
R_1,R_2>\max\bigl\{\|\psi^{\ell}\|_{L_T^\infty H^2},\|\bm{A}^{\ell}\|_{L_T^\infty H^1},\|\bm{A}^{\ell}\|_{W_T^{1,4}L^4}\big| \ell\in\{1,2\}\bigr\}
\end{align*}
and a time $T_*\in(0,T]$ such that $\Phi$ is a contraction on $\mathcal{Z}_{T_*}$. Then the vector field $\bm{B}=\bm{A}^{\ell}|_{\mathcal{I}_{T_*}}\in C\bigl(\mathcal{I}_{T_*};H^{\frac{3}{2}}\bigr)\cap C^1\bigl(\mathcal{I}_{T_*};H^{\frac{1}{2}}\bigr)$ solves the initial value problem \eqref{KGE}+\eqref{KGinit} on $\mathcal{I}_{T_*}$ with $(\psi,\bm{A})=\bigl(\psi^{\ell}|_{\mathcal{I}_{T_*}},\bm{A}^{\ell}|_{\mathcal{I}_{T_*}}\bigr)$ so by uniqueness of solutions to the Klein-Gordon initial value problem \cite[Theorem 3.2]{Sogge} we have 
\begin{align*}
\bm{A}^{\ell}|_{\mathcal{I}_{T_*}}(t)=\mathscr{V}_{\frac{4\pi}{c}\sum_{j=1}^N P\bm{J}_j[\psi^{\ell}|_{\mathcal{I}_{T_*}},\bm{A}^{\ell}|_{\mathcal{I}_{T_*}}]+\bm{A}^{\ell}|_{\mathcal{I}_{T_*}}}(t,0)[\bm{A}_0,\bm{A}_1]
\end{align*}
for $t\in\mathcal{I}_{T_*}$. We conclude that $\bm{A}^1|_{\mathcal{I}_{T_*}},\bm{A}^2|_{\mathcal{I}_{T_*}}\in W^{1,4}\bigl(\mathcal{I}_{T_*};L^4(\mathbb{R}^3;\mathbb{R}^3)\bigr)$ by Lemma \ref{KleinGordon}. Likewise, for $\ell\in\{1,2\}$ the map $\xi=\psi^{\ell}|_{\mathcal{I}_{T_*}}\in C(\mathcal{I}_{T_*};H^2)$ solves the initial value problem \eqref{LinSchr}+\eqref{LinSchrI} on $\mathcal{I}_{T_*}$ with $\bm{A}=\bm{A}^{\ell}|_{\mathcal{I}_{T_*}}$ so Lemma \ref{integralformofequation} gives that 
\begin{align*}
\psi^{\ell}|_{\mathcal{I}_{T_*}}(t)=\mathscr{U}_{\bm{A}^{\ell}|_{\mathcal{I}_{T_*}}}(t,0)\psi_0
\end{align*}
for $t\in\mathcal{I}_{T_*}$. Consequently, $\bigl(\psi^1|_{\mathcal{I}_{T_*}},\bm{A}^1|_{\mathcal{I}_{T_*}}\bigr)$ and $\bigl(\psi^2|_{\mathcal{I}_{T_*}},\bm{A}^2|_{\mathcal{I}_{T_*}}\bigr)$ are both fixed points for the contraction $\Phi:\mathcal{Z}_{T_*}\to \mathcal{Z}_{T_*}$, whereby we must have $\bigl(\psi^1|_{\mathcal{I}_{T_*}},\bm{A}^1|_{\mathcal{I}_{T_*}}\bigr)=\bigl(\psi^2|_{\mathcal{I}_{T_*}},\bm{A}^2|_{\mathcal{I}_{T_*}}\bigr)$.
\end{proof}
In fact, Lemma \ref{uniq1} holds true with $T_*=T$.
\begin{lemma}\label{uniq2}
Given $(\psi_0,\bm{A}_0,\bm{A}_1)\in H^2(\mathbb{R}^{3N})\times H^{\frac{3}{2}}(\mathbb{R}^3;\mathbb{R}^3)\times H^{\frac{1}{2}}(\mathbb{R}^3;\mathbb{R}^3)$ with $\mathrm{div}\bm{A}_0=\mathrm{div}\bm{A}_1=0$ and $T>0$ there exists at most one pair $(\psi,\bm{A})$ in the space $C(\mathcal{I}_{T};H^2)\times \bigl(C\bigl(\mathcal{I}_T;H^{\frac{3}{2}}\bigr)\cap C^1\bigl(\mathcal{I}_T;H^{\frac{1}{2}}\bigr)\bigr)$ that solves \eqref{MSPmaerke}$+$\eqref{initialcondi} and satisfies $\mathrm{div}\bm{A}(t)=0$ for all $t\in\mathcal{I}_T$.
\end{lemma}
\begin{proof}
For $T>0$ and $(\psi_0,\bm{A}_0,\bm{A}_1)\in H^2(\mathbb{R}^{3N})\times H^{\frac{3}{2}}(\mathbb{R}^3;\mathbb{R}^3)\times H^{\frac{1}{2}}(\mathbb{R}^3;\mathbb{R}^3)$ with $\mathrm{div}\bm{A}_0=\mathrm{div}\bm{A}_1=0$ consider two solutions $(\psi^1,\bm{A}^1)$ and $(\psi^2,\bm{A}^2)$ to \eqref{MSPmaerke}+\eqref{initialcondi} that belong to $C(\mathcal{I}_{T};H^2)\times \bigl(C\bigl(\mathcal{I}_T;H^{\frac{3}{2}}\bigr)\cap C^1\bigl(\mathcal{I}_T;H^{\frac{1}{2}}\bigr)\bigr)$ and satisfy $\mathrm{div}\bm{A}^1(t)=\mathrm{div}\bm{A}^2(t)=0$ for all $t\in\mathcal{I}_T$. Then the continuity of the mappings $\psi^1,\psi^2:\mathcal{I}_T\to H^2$, $\bm{A}^1,\bm{A}^2:\mathcal{I}_T\to H^{\frac{3}{2}}$ and $\partial_t\bm{A}^1,\partial_t\bm{A}^2:\mathcal{I}_T\to H^{\frac{1}{2}}$ gives that the number
\begin{align*}
t_0=\sup\bigl\{t\in[0,T]\, \big|\, (\psi^1,\bm{A}^1)=(\psi^2,\bm{A}^2)\textrm{ on }[0,t]\bigr\}
\end{align*}
satisfies
\begin{align*}
(\psi^1(t_0),\bm{A}^1(t_0),\partial_t\bm{A}^1(t_0))=(\psi^2(t_0),\bm{A}^2(t_0),\partial_t\bm{A}^2(t_0)).
\end{align*}
With the intention of reaching a contradiction we assume that $t_0<T$. Then for $\ell\in\{1,2\}$ the pair $\bigl(\widetilde{\psi}^\ell,\widetilde{\bm{A}}^\ell\bigr)\in C(\mathcal{I}_{T-t_0};H^2)\times \bigl(C\bigl(\mathcal{I}_{T-t_0};H^{\frac{3}{2}}\bigr)\cap C^1\bigl(\mathcal{I}_{T-t_0};H^{\frac{1}{2}}\bigr)\bigr)$ given by
\begin{align*}
\bigl(\widetilde{\psi}^\ell(t),\widetilde{\bm{A}}^\ell(t)\bigr)=\bigl(\psi^\ell(t+t_0),\bm{A}^\ell(t+t_0)\bigr)\textrm{ for }t\in\mathcal{I}_{T-t_0}
\end{align*}
takes the initial values
\begin{align*}
\widetilde{\psi}^\ell(0)=\psi^1(t_0), \widetilde{\bm{A}}^\ell(0)=\bm{A}^1(t_0)\textrm{ and }\partial_t\widetilde{\bm{A}}^\ell(0)=\partial_t\bm{A}^1(t_0)
\end{align*}
and satisfies \eqref{MSPmaerke} on $\mathcal{I}_{T-t_0}$. Thus, Lemma \ref{uniq1} gives the existence of some time $T_*\in(0,T-t_0]$ such that $\bigl(\widetilde{\psi}^1,\widetilde{\bm{A}}^1\bigr)$ and $\bigl(\widetilde{\psi}^2,\widetilde{\bm{A}}^2\bigr)$ agree on $[0,T_*]$, whereby the pairs $(\psi^1,\bm{A}^1)$ and $(\psi^2,\bm{A}^2)$ agree on $[t_0,t_0+T_*]$. This contradicts the definition of $t_0$, whereby we can conclude that $(\psi^1,\bm{A}^1)$ and $(\psi^2,\bm{A}^2)$ agree on all of the interval $\mathcal{I}_T$.
\end{proof}

\renewcommand{\bibname}{References}
\bibliographystyle{plain}
\bibliography{Bibliography}
\end{document}